\documentclass[a4paper,10pt]{article}

\usepackage{amsmath, amssymb, amsthm}
\usepackage{hyperref}
\usepackage{color}
\newtheorem{theorem}{Theorem}[section]
\newtheorem*{theorem*}{Theorem}
\newtheorem{proposition}[theorem]{Proposition}
\newtheorem{lemma}[theorem]{Lemma}
\newtheorem{corollary}[theorem]{Corollary}
\newtheorem{definition}[theorem]{Definition}
\theoremstyle{definition}
\newtheorem{remark}[theorem]{Remark}
\newtheorem{example}[theorem]{Example}

\newcommand{\N}{{\mathbb N}}     % Non negative integers
\newcommand{\Z}{{\mathbb Z}}     % Integers
\newcommand{\RR}{{\mathbb R}}    % Reals
\newcommand{\C}{{\mathbb C}}     % Complexe
\newcommand{\suc}{\textit{Suc}}     
\newcommand{\expp}{\textit{exp3}}     
\newcommand{\card}{\textit{card}}     
\newcommand{\lcm}{\textit{lcm}}     
\newcommand{\round}{\textit{round}}

{\end{list}}%

\newenvironment{conditionsabc}
{%
	\begin{list}{\rm (\alph{enumi})}%
	{\noindent%
		\usecounter{enumi}%
		\setlength{\topsep}{2pt}%
		\setlength{\partopsep}{0pt}%
		\setlength{\itemsep}{2pt}%
		\setlength{\parsep}{0pt}%
		\setlength{\leftmargin}{2.5em}%
		\setlength{\labelwidth}{1.5em}%
		\setlength{\labelsep}{0.5em}%
		\setlength{\listparindent}{0pt}%
		\setlength{\itemindent}{0pt}%
	}%
}%
{\end{list}}%

\newenvironment{conditionsiii}
{%
	\begin{list}{\rm (\roman{enumi})}%
	{\noindent%
		\usecounter{enumi}%
		\setlength{\topsep}{2pt}%
		\setlength{\partopsep}{0pt}%
		\setlength{\itemsep}{2pt}%
		\setlength{\parsep}{0pt}%
		\setlength{\leftmargin}{2.5em}%
		\setlength{\labelwidth}{1.5em}%
		\setlength{\labelsep}{0.5em}%
		\setlength{\listparindent}{0pt}%
		\setlength{\itemindent}{0pt}%
	}%
}%
{\end{list}}%

\newenvironment{conditionsbullet}
{%
	\begin{list}{\rm \textbullet}%
	{\noindent%
		\usecounter{enumi}%
		\setlength{\topsep}{2pt}%
		\setlength{\partopsep}{0pt}%
		\setlength{\itemsep}{2pt}%
		\setlength{\parsep}{0pt}%
		\setlength{\leftmargin}{2.5em}%
		\setlength{\labelwidth}{1.5em}%
		\setlength{\labelsep}{0.5em}%
		\setlength{\listparindent}{0pt}%
		\setlength{\itemindent}{0pt}%
	}%
}%
{\end{list}}%

\begin{document}

\title {Newton representation of functions
over natural integers having integral difference ratios}

\author{Patrick C\'egielski%
\footnote{Partially supported by TARMAC ANR agreement
12 BS02 007 01.}%
\newcounter{thanks}
\setcounter{thanks}{\value{footnote}}
\footnote{LACL, EA 4219, Universit\'e Paris-Est Cr\'eteil, France, email : cegielski@u-pec.fr.}
\and
Serge Grigorieff%
\footnotemark[\value{thanks}]
\footnote{LIAFA, CNRS and Universit\'e Paris-Diderot, France, firstname.name@liafa.univ-paris-diderot.fr.}
\newcounter{fnnumber}
\setcounter{fnnumber}{\value{footnote}}
\and
Ir\`ene  Guessarian%
\footnotemark[\value{thanks}]
\footnotemark[\value{fnnumber}]
%\footnote{Corresponding author}
}%

\maketitle

\bibliographystyle{plain}

\begin{abstract}
Different questions lead to the same class of functions
from natural integers to  integers: those which have integral difference ratios,
i.e. verifying $f(a)-f(b)\equiv0 \pmod { (a-b)}$ for all $a>b$.

We characterize this class of functions
via their representations as Newton series.
This class, which obviously contains all polynomials with integral coefficients,
also contains unexpected functions,
for instance all functions $x\mapsto\lfloor e^{1/a}\;a^x\;x!\rfloor$,
with $a\in\Z\setminus\{0,1\}$,
and a function equal to $\lfloor e\;x!\rfloor$ except on $0$.
Finally, to study the complement class, we look at
functions $\N\to\RR$ which are not uniformly close
to any function having integral difference ratios.

\noindent \textbf{Keywords.}
\end{abstract}
%{\footnotesize\tableofcontents}

\noindent{\bf \Large Contents}
\smallskip\\ {\footnotesize
\begin{tabular}{lll}
\bf 1&\multicolumn{2}{l}{\bf\small Introduction}\\
&1.1&What is in the paper?\\
&1.2&Where does the problem come from?\\
&&1.2.1\quad Closure properties of lattices of subsets of $\N$\\
&&1.2.2\quad Uniform continuity properties related to varieties of groups\\
\bf 2&\multicolumn{2}{l}{\bf\small Characterization of
the integral difference ratios property}\\
&2.1&Newton series of functions $\N\to\Z$
and integral difference ratios\\
&2.2&Preparatory lemmas for the proof of Theorem~\ref{thm:main1}\\
&2.3&Proof of Theorem~\ref{thm:main1}\\
\bf3&\multicolumn{2}{l}{\bf\small Examples of functions having
integral difference ratios}\\
&3.1&Main examples: around the factorial function\\
&3.2&Algebra of functions having integral difference ratios
 and applications\\
 &3.3&Examples with generalized hyperbolic functions \\
 &3.4&Asymptotic equivalence\\
 \bf 4&\multicolumn{2}{l}{\bf\small Outside the family of functions
 with integral difference ratios}\\
&4.1&Uniform closeness\\
&4.2&A general negative result for uniform closeness
to functions having integral difference ratios\\
&4.3&Non integral polynomial functions\\
&4.4&Functions around the exponential functions\\
&4.5&Functions around the factorial function
\end{tabular}}

%%%%%%%%%%%%%%%%%%%%%%%%%%%%
%%%%%%%%%%%%%%%%%%%%%%%%%%%%
%%%%%%%%%%%%%%%%%%%%%%%%%%%%
\section{Introduction}
%%%%%%%%%%%%%%%%%%%%%%%%%%%%
%%%%%%%%%%%%%%%%%%%%%%%%%%%%
%%%%%%%%%%%%%%%%%%%%%%%%%%%%
%
%
%%%%%%%%%%%%%%%%%%%%%%%%%%%%
\subsection{What is in the paper?}
%%%%%%%%%%%%%%%%%%%%%%%%%%%%
%
We consider the following question:
what are the functions $f:\N\to\Z$ having integral difference ratios,
i.e. such that $a-b$ always divides $f(a)-f(b)$~?

Our motivation for such functions came from questions in theoretical
computer science, cf. \S\ref{ss:how}.
But these functions are clearly interesting per se.

In \S\ref{s:main1} we characterize functions having
integral difference ratios as the
$\N\to\Z$ functions associated to Newton series such that
the least common multiple of $2,3,\ldots,k$ divides the $k$-th coefficient.

\S\ref{s:examples} is devoted to examples of functions having
integral difference ratios.
Polynomials with coefficients in $\Z$ are trivial examples.
The above characterization shows that there are a lot of non polynomial
examples.
It turns out that some of them are simply expressible.
For instance (cf. \S\ref{ss:e factorial}),
the functions 
$$
x\mapsto\lfloor e^{1/a}\;a^x\;x!\rfloor
\text{\quad with $a\in\Z\setminus\{0,1\}$}
\quad,\quad
x\mapsto\left\{\begin{array}{ll}
1&\text{if\quad}x=0\\
\lfloor e\;x!\rfloor&\text{if\quad}x\in\N\setminus\{0\}
\end{array}\right.
$$
There are also examples of such functions
which oscillate in a periodic way
between several simple expressions of the values, for instance
(cf. \S\ref{ss:e factorial generalized}) the functions
which map $x\in\N$ to
$$
\left\{\begin{array}{ll}
\lfloor\cosh(1/2)\;2^x\;x!\rfloor&\text{if\quad}x\in2\N\\
\lfloor\sinh(1/2)\;2^x\;x!\rfloor&\text{if\quad}x\in2\N+1
\end{array}\right.
\ \,\ \
\left\{\begin{array}{ll}
0&\text{if\quad}x=0\\
\lfloor\cosh(1/2)\;2^x\;x!\rfloor&\text{if\quad}x\in2\N+1\\
\lfloor\sinh(1/2)\;2^x\;x!\rfloor&\text{if\quad}x\in2\N+2
\end{array}\right.
$$
To witness the richness of the family of functions having
integral difference ratios,
we prove (cf. \S\ref{ss:asymptotic}) that
this family contains functions asymptotically equivalent
to large enough functions
(larger than $(2e+\varepsilon)^x$ for some $\varepsilon>0$
suffices).

Finally, in \S\ref{s:out} we show that the above examples
are kind of exceptions: as can be expected,
most functions similar to the above examples {\it do not}
have rational difference ratios. 
Worse, they are not uniformly close to any function
having rational difference ratios. 
In fact, it turns out that proving non uniform closeness
is a very manageable tool to prove failure of the
integral difference ratios property.
First, we use (cf. \S\ref{ss:non idr close}) a classical result
from the theory of uniform distribution modulo one
to get a general result about non uniform closeness:
if $\inf\{|\lambda_x-\lambda_y|\mid x,y\in\N,\ x\neq y\}>0$
then, for almost all real number $\alpha$
(in the sense of Lebesgue measure),
the function $x\mapsto\alpha\lambda_x$ is uniformly close to
no function having integral difference ratios.
Then we look at simple particular classes of functions.
%\begin{conditions}
%\item
\\$\bullet\ $
For non constant polynomials with real coefficients,
we show (cf. \S\ref{ss:poly not idr})
that closeness to a function having integral difference ratios 
holds if and only if all coefficients are in $\Z$
(in which case this polynomial function
has integral difference ratios).
\\$\bullet\ $
For $\alpha\neq0$, all exponential functions $\alpha\; k^x$
(with $k\in\N\setminus\{0,1\}$)
fail to be uniformly close to a function
having integral difference ratios
(cf. \S\ref{ss:exp not idr}).
\\$\bullet\ $
As seen by the examples mentioned supra,
the case of functions $\alpha\,a^x\; x!$ is more delicate.
We study it in \S\ref{ss:fact not idr}).
%\end{conditions}

%%%%%%%%%%%%%%%%%%%%%%%%%%%%
\subsection{Where does the problem come from?}
\label{ss:how}
%%%%%%%%%%%%%%%%%%%%%%%%%%%%
%
A function $f\colon \N\to\Z$ is said to have {\em integral difference ratios}
if\ $\dfrac{f(a)-f(b)}{a-b}\in\Z$\ for all $b<a$.
As far as we know, 
the class of functions $\N\to\Z$ with integral difference ratios
emerged in Pin \& Silva, 2011 \cite{PinSilva}
(see also \S4.2 in \cite{PinSilva2011})
and in our paper \cite{CGG2013}.
In the latter, 
we showed  that the integral difference ratio property characterizes
closure of lattices of regular subsets of $\N$ under inverse image by $f$
(Theorem \ref{thm:CGG1} below).

%----------------------------------------------------------
\subsubsection{Closure properties of lattices of subsets of $\N$.}
%----------------------------------------------------------

\begin{theorem}[\cite{CGG2013}]\label{thm:CGG1}
Let $\suc:\N\to\N$ be the successor function and let $f:\N\to\N$.
The following conditions are equivalent.
\begin{conditionsiii}
\item
The map $f$ is non decreasing and satisfies $f(a)\geq a$ and
has integral difference ratios.
\item
For every finite set $L\subset\N$,
the smallest lattice of subsets of $\N$ containing $L$
and closed under $\suc^{-1}$ is also closed under $f^{-1}$.
\item 
For every arithmetic progression $L=q+r\N$, $q,r\in\N$, $r>0$, 
the smallest lattice of subsets of $\N$ containing $L$
and closed under $\suc^{-1}$ is also closed under $f^{-1}$.
\item
Every lattice of regular subsets of $\N$  %eventually periodic 
which is closed under $\suc^{-1}$ is closed under $f^{-1}$.
\end{conditionsiii}
\end{theorem}

%----------------------------------------------------------
\subsubsection{Uniform continuity properties
related to varieties of groups.}
%----------------------------------------------------------
To state the result in \cite{PinSilva}, we need to recall
some of the involved basic notions
though they are not used anywhere else in the paper.
\begin{definition}
1. A class of finite monoids is a {\rm variety} if it closed
under taking submonoids, quotients and finite direct products.
\\
2. Given a variety $\mathbf{V}$ of finite monoids and a monoid $M$,
the pseudo-metric $d_{\mathbf{V}}:M\times M\to[0,1]$ is defined by
$d_{\mathbf{V}}(x,y)=2^{-k}$ where $k$ is least such that
$\varphi(x)\neq\varphi(y)$ for some morphism $\varphi:M\to F$
such that $F\in\mathbf{V}$ has $k$ elements
(and $k=+\infty$ if there is no such morphism).
\end{definition}

\begin{theorem}[\cite{PinSilva},cf. also \S4.2 in \cite{PinSilva2011})]
\label{thm:pinsilva}
Let $\mathbf{G}$ be the variety of finite groups.
Consider the monoid $(\Z,+)$ and let $f:\Z\to\Z$.
The following conditions are equivalent.
\begin{conditionsiii}
\item
$f$ is $d_{\mathbf{V}}$-uniformly continuous for every subvariety
$\mathbf{V}$ of $\mathbf{G}$.
\item
$u-v$ divides $f(u)-f(v)$ for all $u,v\in\Z$.
\end{conditionsiii}
%Let $V$ be a variety of finite monoids.
%A function is $V$-hereditarily continuous if it is $W$-uniformly continuous
%for each subvariety $W$ of $V$.
%Note that $V$-hereditary continuity is in general a stronger property
%than $V$-uniform continuity.
%For instance, it is shown in [11] that a function $f$ from $Z$ to $Z$
%is $G$-uniformly continuous if and only if, 
%for all $r\in\N$, there exists $s\in\N$ such that for all $u,v\in\Z$,
%$u\equiv v\pmod s$ implies $f(u)\equiv f(v) \pmod r$.
%On the other hand,
%$f$ is $G$-hereditarily continuous if and only if, for all $u,v\in\Z$,
%$u-v$ divides $f(u)-f(v)$.
%It follows that the function
%$f:\Z\to\Z$ defined by $f(2n)=0$, $f(2n+1)=1$
%is $G$-uniform continuous but not $G$-hereditarily continuous.
\end{theorem}
%
%----------------------------------------------------------
%\subsubsection{Future work}{\LARGE\centerline{Je mettrai ce paragraphe dans une conclusion/discussion?intro?}}

%----------------------------------------------------------
%Theorem~\ref{thm:pinsilva} invites to look at the class of functions
%$\Z\to\Z$ with integral difference ratios.
%It turns out that this class does not simply reduce to that
%of functions $\N\to\Z$ which is considered in this paper.
%We are study the $\Z\to\Z$ case.% in another paper in preparation. 

%%%%%%%%%%%%%%%%%%%%%%%%%%%%
%%%%%%%%%%%%%%%%%%%%%%%%%%%%
%%%%%%%%%%%%%%%%%%%%%%%%%%%%
\section{Characterization of the integral difference ratios property}
\label{s:main1}
%%%%%%%%%%%%%%%%%%%%%%%%%%%%
%%%%%%%%%%%%%%%%%%%%%%%%%%%%
%%%%%%%%%%%%%%%%%%%%%%%%%%%%
%
To get a characterization of functions having integral difference ratios,
we use Newton series \cite{Newton1687,Boole1872}, originally introduced to study  functions from $\RR$ to  $\RR$, but here reduced to functions from $\N$ to  $\Z$.
%
%%%%%%%%%%%%%%%%%%%%%%%%%%%%
\subsection{Newton series of functions $\N\to\Z$
and integral difference ratios}
\label{ss:main1}
%%%%%%%%%%%%%%%%%%%%%%%%%%%%
%
Our first result, Theorem~\ref{thm:main1},
involves notions recalled in 
Definition~\ref{def:lcm} and Proposition~\ref{p:Newton} below.

\begin{definition}
[Newton representation for functions $\N\to\Z$]
\label{d:Newton} A map $f:\N\to\Z$ has a {\rm Newton representation}
if there exists some sequence $(a_k)_{k\in\N}$ such that, for all $x\in\N$,
the value $f(x)$ is equal to the sum of the series
\begin{eqnarray}\label{Newton}
f(x)&=&\sum_{k\in \N} a_k\; \frac{\prod_{i=0}^{k-1}(x-i)}{k!}
\;\;=
\;\;\sum_{k\in \N} a_k \dbinom{x}{k}\qquad
\end{eqnarray}
\end{definition}
\begin{proposition}
[Newton series correspondence]
\label{p:Newton}
A bijective correspondence between
sequences $(a_k)_{k\in\N}$ of integers in $\Z$
and maps $f:\N\to\Z$
is set up by the Newton representation \eqref{Newton} where, for $k\in\N$,
\begin{eqnarray}\label{Newton bis}
 a_k&=&\sum_{i=0}^{i=k}(-1)^{k-i}f(i)\dbinom{k}{i}
\end{eqnarray}
\end{proposition}
\begin{proof} 
Observe that, for every $x\in\N$, 
the binomial coefficient $\dbinom{x}{k}$ is null for $k>x$, 
hence the infinite series defining $f(x)$ in \eqref{Newton}
reduces to a finite sum for any given non negative $x$.
This removes any convergence problem.
Since the binomial coefficients are in $\N$,
for every sequence $(a_k)_{k \in\N}\in\Z^\N$,
equation \eqref{Newton} represents a map from $\N$ into $\Z$.

Conversely, every $f:\N\to\Z$ has a unique such representation
since \eqref{Newton} insures that
$$
a_0 = f(0)\quad,\quad
a_{k+1} = f(k+1) - \sum_{i=0}^{i=k} a_i\dbinom{k+1}{i} \in \Z\ .
$$
Inverting the binomial lower triangular matrix
$B_k=(\dbinom{i}{j})_{0\leq i,j\leq k}$
in equality
$$
\begin{pmatrix}
f(0)\\f(1)\\f(2)\\f(3)\\\vdots\\f(k)
\end{pmatrix}
=\begin{pmatrix}
1\\
1&1\\
1&2&1\\
1&3&3&1\\
\vdots&\vdots&\vdots&\cdots&\vdots&\\
1&\dbinom{k}{1}&\dbinom{k}{2}&\cdots&\dbinom{k}{k-1}&1
\end{pmatrix}
\;\begin{pmatrix}
a_0\\a_1\\a_2\\a_3\\\vdots\\a_k
\end{pmatrix}
$$
we get formula \eqref{Newton bis}
since the inverse of $B_k$ is the triangular matrix
$\widetilde{B}_k=((-1)^{i-j}\dbinom{i}{j})_{0\leq i,j\leq k}$
(cf. for instance \cite{YangLiu2006}).
Indeed, the $(i,j)$ element of $B_k\widetilde{B}_k$ is
$\sum_{\ell=0}^{\ell=k}\dbinom{i}{\ell}(-1)^{\ell-j}\dbinom{\ell}{j}$.
Recall that (cf \cite{GKP} page 174)
$\dbinom{i}{\ell}\dbinom{\ell}{j}
=\dbinom{i}{ j}\dbinom{i-j} {\ell-j}$. Hence,
$$\sum_{\ell=0}^{\ell=k}\dbinom{i}{\ell}(-1)^{\ell-j}\dbinom{\ell}{j}
=\dbinom{i}{j}\sum_{\ell=0}^{\ell=k}(-1)^{\ell-j}\dbinom{i-j}{\ell-j}
=\dbinom{i}{j}(1-1)^{i-j}
=\left\{\begin{array}{ll}
0&\text{if }i\neq j\\
1&\text{if }i=j 
\end{array}
\right.
$$which proves that $\widetilde{B}_k$ is the inverse of $B_k$.
\end{proof}
To state the main Theorem of the present section,
we need to recall another classical notion.
\begin{definition}\label{def:lcm}
For $k\in\N$, $k\geq1$, $\lcm(k)$ is the least common multiple of
all positive integers less than or equal to $k$.
By convention, $\lcm(0)=1$.
\end{definition}
\begin{remark}\label{rk:lcm}
The Neperian logarithm of the $\lcm$ function was introduced by
Chebychev, 1852 \cite{chebyshev1852}:
letting $\ell(p,x)=\lfloor\log_p(x)\rfloor$
be the greatest integer $k$ such that $p^k\leq x$,
$$
\psi(x)
=\sum\{\ell(p,x)\log p \mid p\leq x,\, \text{$p$ prime}\}
=\log(\lcm(x))\ .
$$
A variant of the prime number theorem insures that
the Chebychev function $\psi(x)$ is asymptotically equivalent to $x$,
i.e. $\lim_{n\to+\infty}\psi(x)/x=1$.
Thus, for any $\varepsilon>0$, we have
$\log(\lcm(x)) = x\,(1+o(x))$ hence
$\lcm(x)=e^{x\,(1+o(x))}=(e^{1+o(x)})^x=(e+o(x))^x$,
i.e. for every $\varepsilon>0$, for all $x$ large enough,
\begin{equation}\label{eq:lcm exp}
(e-\varepsilon)^x\leq \lcm(x)\leq (e+\varepsilon)^x\ .
\end{equation}
Simple lower and upper bounds of $\lcm$ are known:
$2^n\leq\lcm(n)$ for $n\geq7$
(cf. formula (9) in \cite{Nair1982} for $n\geq9$
plus direct check for $n=7,8$)
and $\lcm(n)<3^n$ for all $n\in\N$
(cf. \cite{hanson1972}).

It is known (\cite{schmidt1903}) that $\psi(x)$ (resp. $\lcm(x)$)
oscillates around $x$ (resp. $e^x$)~:
for some $K>0$,
there are infinitely many $x$'s such that $\lcm(x)<e^{x-K\sqrt{x}}$
and infinitely many $x$'s such that $\lcm(x)>e^{x+K\sqrt{x}}$.
Nevertheless, the relative oscillation is in $o(e^x)$
(\cite{RosserSchoenfeld1975} Theorem 8):
letting $a=2\times10^7$, for all $x\geq2$,
$$
e^{x-a(x/\log^4 x)}<\lcm(x)<e^{x+a(x/\log^4 x)}\ .
$$
Assuming Riemann's hypothesis, a better approximation is possible
(\cite{schoenfeld1976} Theorem 10):
$$
\text{Riemann's hypothesis proves} \ \
e^{x-(\sqrt{x}\log^2(x)/8\pi)}<\lcm(x)<e^{x+(\sqrt{x}\log^2(x)/8\pi)}\ .
$$
For recent results around the $\lcm$ function, see
\cite{QianHong2013,hong-et-al2011,farhi2007,FarhiKane2009,dusart2010}.
\end{remark}
\begin{theorem} \label{thm:main1}
Let $f:\N\to\Z$ be a function with Newton representation
$\sum_{k\in \N} a_k\,  \dbinom{x}{k}$.
The following conditions are equivalent:
\begin{conditionsiii}
\item
$f$ has integral difference ratios,
i.e. $\dfrac{f(a)-f(b)}{a-b}\in\Z$ for all $b\neq a$.
\item
$\lcm(k)$ divides $a_k$ for all $k\in\N$.
\end{conditionsiii}
\end{theorem}
\begin{proof}
See \S\ref{ss:preparatory}, \ref{sss:proof main1 i to ii},
\ref{sss:proof main1 ii to i}.
\end{proof}
We now state a corollary whose proof does not need
the machinery of the proof of Theorem~\ref{thm:main1}.
\begin{corollary}\label{rk:main1}
If $k!$ divides $a_k$ for all $k\in\N$ 
then $f$ has integral difference ratios.
\end{corollary}
\begin{proof}
Let $a_k=k!\;b_k$ (with $b_k\in\Z$);
then
{$f(x)\;=\;\sum_{k\in \N} a_k \dbinom{x}{k}
\;=\;b_0+b_1x+b_2x(x-1)+b_3x(x-1)(x-2)+\cdots$.}
For $a,b\in\N$ and $N>\max(a,b)$,
$f(a)-f(b)$ is the difference of the values on $a,b$ of the polynomial
$\sum_{k<N}b_k\prod_{i=0}^{i=k-1}(x-i)$ which has coefficients in $\Z$, hence $a-b$
divides $f(a)-f(b)$.
\end{proof}

%
%%%%%%%%%%%%%%%%%%%%%%%%%%%%
\subsection{Preparatory lemmas for the proof of Theorem~\ref{thm:main1}}
\label{ss:preparatory}
%%%%%%%%%%%%%%%%%%%%%%%%%%%%
%
The proof of Theorem~\ref{thm:main1} relies on three lemmas whose proofs are elementary.
\begin{lemma}\label{lemme.p}
If $0\leq n-k<p\leq n$ then $p$ divides $\lcm(k)\dbinom{n}{k}$.
\end{lemma}
\begin{proof} By induction on $n\geq1$.
The initial case $n=1$ is trivial since condition $0\leq n-k<p\leq n$ yields
$p=k=1$.
Induction step: assuming the result for $n$, we prove it for $n+1$.
Suppose $0\leq n+1-k<p\leq n+1$.
\\
{\it First case: $p\leq n$.}
Then we have $0\leq n-k<p\leq n$ and $0\leq n-(k-1)<p\leq n$,
so that, by induction hypothesis,
$p$ divides $\lcm(k)\dbinom{n}{k}$
and $p$ divides $\lcm(k-1)\dbinom{n}{k-1}$.
A fortiori, $p$ divides
$\lcm(k)\dbinom{n}{k}+\lcm(k)\dbinom{n}{k-1}
=\lcm(k)\dbinom{n+1}{k}$
(by Pascal's formula).
\\
{\it Second case: $p=n+1$.}
Then $k\geq1$ and
$\lcm(k)\dbinom{n+1}{k}=(n+1)\frac{\lcm(k)}{k}\dbinom{n}{k-1}$
hence $p=n+1$ divides
$\lcm(k)\dbinom{n+1}{k}$.
\end{proof}
\begin{lemma}\label{l:n divise Ankb}
If $n,k,b\in\N$ and $k\leq b$ then $n$ divides
$A^n_{k,b}=\lcm(k) \left(\dbinom{b+n}{k}-\dbinom{b}{k}\right)$.
\end{lemma}
\begin{proof}
We argue by double induction on $k$ and $b$ with the conditions
$$
(\+P_{k,b})\quad \forall n\in\N,\;\;n\text{ divides }A^n_{k,b}
\qquad,\qquad
(\+P_k)\quad \forall b\geq k,\;\forall n\in\N,\;\;n\text{ divides }A^n_{k,b}\;.
$$
Conditions $(\+P_0)$ and $(\+P_1)$ are trivial
since $A^n_{0,b}=0$ and $A^n_{1,b}=n$.

Suppose $k\geq1$ and $(\+P_k)$ is true.
To prove $(\+P_{k+1})$, we prove by induction on $b\geq k+1$
that $(\+P_{k+1,b})$ holds.

In the basic case $b=k+1$, we have
\begin{eqnarray*}
A^n_{k+1,k+1} &=&
\lcm(k+1)\;\left(\dbinom{k+1+n}{k+1} - \dbinom{k+1}{k+1}\right)
\\
&=&\lcm(k+1)\;\left(\dbinom{k+n}{k} + \dbinom{k+n}{k+1} - 1\right)
\quad {\hbox{\rm by Pascal's relation}}
\\
&=&\lcm(k+1)\;\left(\dbinom{k+n}{k} - \dbinom{k}{k}\right)
     + \lcm(k+1)\;\dbinom{k+n}{k+1}
\\
&=&\frac{\lcm(k+1)}{\lcm(k)}\;A^n_{k,k}
     + \lcm(k+1)\;\dbinom{k+n}{k+1}
\end{eqnarray*}
Since $(\+P_{k,k})$ holds (induction hypothesis on $k$),
$n$ divides $A^n_{k,k}$ hence divides the first term.
If $n\leq k+1$ then $n$ divides $\lcm(k+1)$ hence divides the second term.
If $n>k+1$, applying Lemma \ref{lemme.p} with $n'=k+n$, $p'=n$ and $k'=k+1$,
shows  that $n=p'$ divides the second term.
Thus, $n$ divides $A^n_{k+1,k+1}$ and $(\+P_{k+1,k+1})$ holds.

Suppose now that $(\+P_{k+1,c})$ holds for $k+1\leq c\leq b$.
We prove $(\+P_{k+1,b+1})$.
Using Pascal's relation, we get
\begin{eqnarray*}
A^n_{k+1,b+1} &=& \lcm(k+1)\;\left(\dbinom{b+1+n}{k+1} - \dbinom{b+1}{k+1}\right)
\\
&=&\lcm(k+1)\;\left(\dbinom{b+n}{k} + \dbinom{b+n}{k+1} 
                              - \dbinom{b}{k} - \dbinom{b}{k+1}\right)
\\
&=&\lcm(k+1)\;\left(\left(\dbinom{b+n}{k} - \dbinom{b}{k}\right)
+ \left(\dbinom{b+n}{k+1} - \dbinom{b}{k+1}\right)\right)
\\
&=&\left(\frac{\lcm(k+1)}{\lcm(k)}\;A^n_{k,b}\right) + A^n_{k+1,b}
\end{eqnarray*}
Since $(\+P_{k,b})$ and $(\+P_{k+1,b})$ hold, $n$ divides both terms
of the above sum hence $n$ divides $A^n_{k+1,b+1}$
and $(\+P_{k+1,b+1})$ holds.
\end{proof}

The following is an immediate consequence of Lemma \ref{l:n divise Ankb}
(set $a=b+n$).
\begin{lemma}\label{lemme.a-b}
If $a\geq b $ then $a-b$ divides
$\lcm(k) \left(\dbinom{a}{k}-\dbinom{b}{k}\right)$
for all $k\leq b$.
\end{lemma}

%%%%%%%%%%%%%%%%%%%%%%%%%%%%
\subsection{Proof of Theorem~\ref{thm:main1}}
\label{ss:proof main1}
%%%%%%%%%%%%%%%%%%%%%%%%%%%%
%
%----------------------------------------------------------------
\subsubsection{Proof of {\rm(i)} $\Rightarrow$  {\rm (ii)}}
\label{sss:proof main1 i to ii}
%----------------------------------------------------------------
%
We suppose that
$f(x)=\sum_{k\in \N} a_k \dbinom{x}{k}$
has integral difference ratios and we show that
$\lcm(k)$ divides $a_k$ for all $k\in\N$.
\smallskip\\
{\it{\normalfont\bf Claim 1.} For all $k\geq1$, $k$ divides $a_k$.}
\\
The proof is by induction.
Recall $f(k)= \sum_{i= 0}^k \dbinom{k}{i} a_i$.\\
{\it Induction Basis:}
The case $k=1$ is trivial.
For $k=2$, observe that $2$ divides $f(2)-f(0)= 2a_1 +a_2$
hence $2$ divides $a_2$.
\\
{\it Inductive Step:} assuming that $\ell$ divides $a_\ell$ for every $\ell\leq k$,
we prove that $k+1$ divides $a_{k+1}$.
Observe that
\begin{eqnarray*}
f(k+1) -f(0)&=&
(k+1) a_1 + \left(\sum_{i=2}^k \dbinom{k+1}{i} a_i\right)+a_{k+1}
\\
&=&(k+1) a_1 + 
\left(\sum_{i=2}^k (k+1)\;\frac{a_i}{i}\;\dbinom{k}{i-1}\right)
+a_{k+1}
\end{eqnarray*}
By the induction hypothesis, ${a_i\over i}$ is an integer for $i\leq k$.
Since $f$ has integral difference ratios, $k+1$ divides $f(k+1) -f(0)$
hence $k+1$ divides the last term $a_{k+1}$ of the sum.
\smallskip\\
{\it{\normalfont\bf Claim 2.} For all $1\leq p\leq k$, $p$ divides $a_k$.
Hence, $\lcm(k)$ divides $a_k$.}
\\
The case $p=1$ is trivial.
We use induction on $p\geq2$.
\\
\textbullet\;{\it Basic case: $2$ divides $a_k$ for all $k\geq2$.}
We argue by induction on $k\geq2$.
Claim~1 gives the base case: $2$ divides $a_2$.
Induction step: assuming that $2$ divides $a_i$ for all $2\leq i\leq k$
we prove that $2$ divides $a_{k+1}$. Two cases can occur.
\\
{\it Subcase 1 : $k+1$ is odd.}
Then 2 divides $f(k+1)-f(1)$.
Now,
\\\centerline{$f(k+1)-f(1)
=ka_1+\left(\sum_{i=2}^k a_i {{k+1}\choose i}\right) +a_{k+1}$,} 
$k$ is even and $2$ divides the $a_i$ for $2\leq i\leq k$ by the induction hypothesis,
hence, $2$ divides $a_{k+1}$.
\\
{\it Subcase 2 : $k+1$ is even.}
Then 2 divides $f(k+1)-f(0)$.
Now,
\\\centerline{$f(k+1)-f(0)
=(k+1)a_1+\left(\sum_{i=2}^k a_i {{k+1}\choose i} \right)+a_{k+1}$,}
$k+1$ is even and $2$ divides the $a_i$ for $2\leq i\leq k$ by the induction hypothesis,
thus, $2$ divides $a_{k+1}$.
\\
\textbullet\;{\it Induction step: assuming that
$p\geq2$ and, for all $q\leq p$, $q$ divides $a_\ell$ for all $\ell\geq q$,
we prove that $p+1$ divides $a_k$ for all $k\geq p+1$.}
\\
We use induction on $k\geq p+1$.
Claim~1 gives the base case: $p+1$ divides $a_{p+1}$.
Induction step: assuming that $p+1$ divides $a_i$ for all $p+1\leq i\leq k$
we prove that $p+1$ divides $a_{k+1}$.
Since $f$ has integral difference ratios, $p+1$ divides $f(k+1)-f(k-p)$
which is equal to
$$
\sum_{i=1}^{i=k-p} a_i \left({{k+1}\choose i}- {{k-p}\choose i}\right)
\;+\
\left(\sum_{i=k+1-p}^{i=k} a_i {{k+1}\choose i}\right)
\;+\;a_{k+1}
$$

Let us first look at the terms of the first sum corresponding to $1\leq i\leq p$.
The induction hypothesis (on $p$) insures that $q$ divides $a_k$
for all $q\leq p$ and $k\geq q$.
In particular, (letting $k=i$) $\lcm(i)$ divides $a_i$.
Since $(k+1)-(k-p)=p+1$, Lemma~\ref{l:n divise Ankb} insures that
$p+1$ divides $\lcm(i)\;\left({{k+1}\choose i}- {{k-p}\choose i}\right)$.
A fortiori, 
$p+1$ divides $a_i\; \left({{k+1}\choose i}- {{k-p}\choose i}\right)$.

We now turn to the terms of the first sum corresponding to
$p+1\leq i\leq k-p$ (if there are any).
The induction hypothesis (on $k$) insures that
$p+1$ divides $a_i$ for all $p+1\leq i\leq k$

Thus, each term of the first sum is divisible by $p+1$.

Consider now the terms of the second sum.
By the induction hypothesis (on $k$),
$p+1$ divides $a_i$ for all $p+1\leq i\leq k$.
%A fortiori, $p+1$ divides the terms associated with the $i$'s
%such that $\max(p+1,k+1-p)\leq i\leq k$.
It remains to look at the terms associated with the $i$'s
such that $k+1-p\leq i\leq p$
(there are such $i$'s in case $k+1-p<p+1$).
For such $i$'s we have $0\leq (k+1)-i\leq (k+1)-p<p+1\leq k+1$
and Lemma \ref{lemme.p}
(used with $k+1,i,p+1$  in place of $n,k,p$)
insures that $p+1$ divides $\lcm(i) {k+1\choose i}$.
Now, for such $i$'s, the induction hypothesis (on $p$)
insures that $\lcm(i)$ divides $a_i$.
Thus, $p+1$ divides $a_i {k+1\choose i}$.

Since $p+1$ divides each one of these three sums,
it must divide the last summand $a_{k+1}$.

This finishes the proof of Claim 2 hence of implication $(i)\Rightarrow (ii)$
in Theorem~\ref{thm:main1}. 
\hfill{$\Box$}

%----------------------------------------------------------------
\subsubsection{Proof of {\rm(ii)} $\Rightarrow$ {\rm(i)} }
\label{sss:proof main1 ii to i}
%----------------------------------------------------------------
%
If $f$ satisfies (ii), then it can be written in the form
$f(n)= \sum_{k= 0}^n b_k \lcm(k){n \choose k}$.
Consequently,
$$
f(a)-f(b)= \left(\sum_{k= 0}^b 
                         b_k \lcm(k)\Big({a \choose k} -{b\choose k}\Big)\right)
+ \sum_{k= b+1}^a b_k \lcm(k){a \choose k}
$$
By Lemma \ref{lemme.a-b}, $a-b$  divides each term of the first sum.

Consider the terms of the second sum.

If $a-b\leq b+1$ then $a-b$ divides each term of the second sum
since $a-b$ divides $\lcm(k)$ for every $k\geq (a-b)$.

If $a-b> b+1$ then, for $b+1\leq k\leq a$ we have $0\leq a-k<a-b\leq a$
and Lemma \ref{lemme.p}
(used with $a,k,a-b$  in place of $n,k,p$)
insures that $a-b$ divides $\lcm(k){a \choose k}$.
Again,  $a-b$ divides each term of the second sum.
\hfill{$\Box$}
%
%

%
%
%%%%%%%%%%%%%%%%%%%%%%%%%%%%%%%
%%%%%%%%%%%%%%%%%%%%%%%%%%%%%%%
%%%%%%%%%%%%%%%%%%%%%%%%%%%%%%%
%%%%%%%%%%%%%%%%%%%%%%%%%%%%%%%
\section{Examples of functions having integral difference ratios}
\label{s:examples}
%%%%%%%%%%%%%%%%%%%%%%%%%%%%%%%
%%%%%%%%%%%%%%%%%%%%%%%%%%%%%%%
%%%%%%%%%%%%%%%%%%%%%%%%%%%%%%%
%%%%%%%%%%%%%%%%%%%%%%%%%%%%%%%

Back to the motivations given in \S\ref{ss:how},
it may not be obvious to find functions $f$ such that for every finite set $L\subset\N$,
the smallest lattice of subsets of $\N$ containing $L$
and closed under $\suc^{-1}$ is also closed under $f^{-1}$.
Our characterization by integral difference ratios (Theorem~\ref{thm:CGG1})
gives a first simple class of such functions:  polynomial functions.
Now, are there non polynomial such functions
expressible with usual mathematical functions?
It turns out that this is the case and can be proved using the characterization
given by Theorem~\ref{thm:main1}:
for instance, the function such that $f(0)=1\ ,\
f(x)=\lfloor e\ x!\rfloor \textit{\ for\ } x\geq1$ (Theorem \ref{thm:e factorial})
and variations thereof (e.g. Corollary \ref{thm:e factorialbis}).

%%%%%%%%%%%%%%%%%%%%%%%%%%%%
\subsection{Main examples: around the factorial function}
\label{ss:e factorial}
%%%%%%%%%%%%%%%%%%%%%%%%%%%%
%
A simple application of Corollary~\ref{rk:main1}
gives functions $\N\to\Z$ having integral difference ratios
with unexpectedly simple analytic expressions
up to the ceil and floor functions $\RR\to\Z$.
\begin{theorem}\label{thm:e factorial}
Let $e$ be the usual Neper constant.
For\ $a\in\Z\setminus\{0,1\}$,
the following functions $\N\to\Z$ have integral difference ratios:
$$
\begin{array}{lllllllll}
\phi_a^-\colon x&\mapsto&\lfloor e^{1/a}\, a^x \ x!\rfloor
&\  \phi_a^+\colon x&\mapsto&\lceil e^{1/a}\, a^x \ x!\rceil
\medskip\\
\phi_1^-\colon x&\mapsto&\!\!\!\left\{\begin{array}{ll}
1&\text{if $x=0$}\\
\lfloor e \ x!\rfloor&\text{if $x\in\N\setminus\{0\}$}
\end{array}\right.
&\  \phi_1^+\colon x&\mapsto&\!\!\!\left\{\begin{array}{ll}
2&\text{if $x=0$}\\
\lceil e \ x!\rceil&\text{if $x\in\N\setminus\{0\}$}
\end{array}\right.
\end{array}
$$

\end{theorem}
\begin{remark}
Functions\ $\lfloor e\;x! \rfloor$ and $\lceil e \ x!\rceil$
do {\em not} have  integral difference ratios
(cf. Proposition~\ref {p:alpha fact not idr}).
\end{remark}
\begin{proof}
Recall Taylor-Lagrange formula applied to $t\mapsto e^t$
(considered as a map on $\RR$): for all $t\in\RR$,
$$
e^t= \left(\frac{1}{0!}
                    +\frac{t}{1!}
                    +\frac{t^2}{2!}
                    +\cdots
                     +\frac{t^{k-1}}{(k-1)!}
                    + \frac{t^k}{k!}\right)
                    + e^{\theta\, t}\,\dfrac{t^{k+1}}{(k+1)!}
$$
for some $0<\theta<1$ depending on $k$ and $t$.

For $a\in\Z$, let $f_a:\N\to\Z$ be the function associated to the Newton
series
\begin{equation}\label{eq:fa}
f_a(x)\ =\ \sum_{n\in\N} a^n\;n!\;\dbinom{x}{n}
\end{equation}
Corollary~\ref{rk:main1} insures that $f_a$ has integral difference ratios.
Moreover,
\begin{eqnarray}\notag
f_a(x)&=&\sum_{k\in \N} a^n\;n!\;\dbinom{x}{n}
\\\notag
&=&a^x\;x!\;\left(\frac{(1/a)^{x}}{x!}
                    +\frac{(1/a)^{x-1}}{(x-1)!}
                    +\cdots
                    +\frac{(1/a)^{1}}{1!}
                    +\frac{(1/a)^{0}}{0!} \right)
\\\notag
&=&a^x\,x!\;\left(e^{1/a} - e^{\theta/a}\,\dfrac{(1/a)^{x+1}}{(x+1)!}\right)
\qquad\text{for some $0<\theta<1$}
\\\label{eq:floor fa}
&&\text{hence}\qquad\qquad
e^{1/a}\,a^x\;x!\ =\ f_a(x) + \dfrac{e^{\theta/a}}{a\,(x+1)}
\end{eqnarray}

{\it Case $a\geq2$.} For $x\in\N$, we have
$0<e^{\theta/a}/(a\,(x+1))<e^{1/2}/2 <1$
and, since $f_a(x)\in\N$, equation~\eqref{eq:floor fa} yields
$f_a(x)=\lfloor e^{1/a}\,a^x\;x! \rfloor$ and
$f_a(x)+1=\lceil e^{1/a}\,a^x\;x! \rceil$.

{\it Case $a\leq-1$.} For $x\in\N$, we have\
$\left|\dfrac{e^{\theta/a}}{a\,(x+1)}\right|
= \dfrac{e^{-\theta/|a|}}{|a|\,(x+1)}
\leq e^{-\theta/|a|} <1$\
and\ $-1< \dfrac{e^{\theta/a}}{a\,(x+1)} <0$.
Since $f_a(x)\in\Z$, equation~\eqref{eq:floor fa} yields
$f_a(x)=\lceil e^{1/a}\,a^x\;x! \rceil$
and $f_a(x)-1=\lfloor e^{1/a}\,a^x\;x! \rfloor$.

{\it Case $a=1$.} For $x\in\N$, $x\geq2$, we have
$0<e^\theta / (x+1)<e/3<1$ and, again, equation~\eqref{eq:floor fa} yields
$f_1(x)=\lfloor e\;x! \rfloor$ and $f_1(x)+1=\lceil e\;x! \rceil$.
Also, $f_1(0)=1<2=\lfloor e\;0! \rfloor$, $f_1(1)=2=\lfloor e\;1!\rfloor$
and $f_1(0)+1=2<3=\lceil e\;0! \rceil$, $f_1(1)+1=3=\lceil e\;0! \rceil$.

Thus, the functions in the statement of the theorem are among
the $f_a$'s, $f_a+1$'s and $f_a-1$'s,
all of which have rational difference ratios.
\end{proof}

%%%%%%%%%%%%%%%%%%%%%%%%%%%%
\subsection{Algebra of functions having integral difference ratios
and applications}
\label{ss:algebra idr}
%%%%%%%%%%%%%%%%%%%%%%%%%%%%
%
In order to get variations of Theorem \ref{thm:e factorial}
we state some closure properties
of the family of functions with integral difference ratios:
sum, product (i.e. they form a subring of functions from $\N$ to $\Z$)
and composition.
\begin{proposition}\label{p:closure sum} {\rm [Subring]}
If $f,g:\N\to\Z$ have integral difference ratios then so have
their sum and product.
\end{proposition}
\begin{proof}
For product, use equality
$f(x)g(x)-f(y)g(y)=f(x)\;(g(x)-g(y)) + g(y)\;(f(x)-f(y))$.
\end{proof}
\begin{corollary}\label{cor:poly}
Every polynomial with coefficients in $\Z$ defines
a function $\N\to\Z$ having integral difference ratios.
\end{corollary}
\begin{proof}
Observe that the identity and constant functions have integral difference ratios
and apply Proposition~\ref{p:closure sum}.
\end{proof}
\begin{corollary}\label{thm:e factorialbis}
Let $s,a\in\Z$, $a\neq0$.
Let $h_{a,s}$ be any one of the functions 
$\lfloor s\,e^{1/a}\, a^x \ x!\rfloor$,
$\lceil s\,e^{1/a}\, a^x \ x!\rceil$, with $a\in\Z\setminus\{0\}$.
There exists a function $g_{a,s}:\N\to\Z$ having integral difference ratios
such that $h_{a,s}(x)=g_{a,s}(x)$ for all $x\geq se-1$.
\end{corollary}
\begin{proof}
Let $f_a$ be as in the proof of Theorem~\ref{thm:e factorial}.
Proposition~\ref{p:closure sum} insures that the function
$g_{a,s}(x) = s\,f_a = s\,\sum_{k\in \N} a^k\;k!\;\dbinom{x}{k}$
has integral difference ratios.
Also, Equation~\eqref{eq:fa} above yields
$$
s\,e^{1/a}\,a^x\;x!
\ =\  s\,f_a(x) + \dfrac{s\,e^{\theta/a}}{a\,(x+1)}
$$
with $0<\theta<1$.
If $x\geq se-1$ then $|s\,e^{\theta/a}/a\,(x+1)!| <1$
and we can argue as in the proof of Theorem~\ref{thm:e factorial}
to finish the proof.
\end{proof}
\begin{example}
The  bound $se-1$ (of Corollary \ref{thm:e factorialbis},
obtained by majorizing $\theta$ by $1$)
may not be optimal.
For instance,
equalities $s\,f_1(x)=s\;\lfloor e\ x!\rfloor=\lfloor s e\ x!\rfloor$
may hold for some $x< se-1$. For instance,
$2e-1=4.436\ldots$, $3e-1=7.154\ldots$ but
\\\centerline{
$\begin{array}{|l|}
\hline
f_2(x)=2\lfloor e\ x!\rfloor
=\lfloor 2e\ x!\rfloor\quad\text{for $x\geq2$}
\\
f_2(0)=2<2\lfloor e\ 0!\rfloor=4<\lfloor 2e\ 0!\rfloor=5
\\
f_2(1)=4=2\lfloor e\ 1!\rfloor<\lfloor 2e\ 1!\rfloor=5
\\\hline
\end{array}$
\ 
$\begin{array}{|l|}
\hline
f_3(x)=3\lfloor e\ x!\rfloor
=\lfloor 3e\ x!\rfloor\quad\text{for $x\geq3$}
\\
f_3(0)=3<3\lfloor e\ 0!\rfloor=6<\lfloor 3e\ 0!\rfloor=8
\\
f_3(1)=6=3\lfloor e\ 1!\rfloor<\lfloor 3e\ 1!\rfloor=8
\\
f_3(2)=15=3\lfloor e\ 2!\rfloor<\lfloor 3e\ 2!\rfloor=16
\\\hline
\end{array}$}
\end{example}

\medskip
Closure under composition gives more variations of our main example (Theorem \ref{thm:e factorial}).
\begin{proposition}\label{p:closure composition} {\rm [Composition]}
 If $f:\N\to\Z$ and $g:\N\to\N$ have integral difference ratios
then so has $f\circ g$.
\end{proposition}
\begin{proof}
Use transitivity of divisibility:
$x-y$ divides $g(x)-g(y)$ which divides $f(g(x))-f(g(y))$.
\end{proof}
The following simple result allows to use
Proposition~\ref{p:closure composition},
to extend the scope of Theorem~\ref{thm:e factorial}.
\begin{proposition}\label{p:bounded}
If $f:\N\to\Z$ is not a constant function and has integral difference ratios
then, for every $z\in\N$, the set $f^{-1}(z)$ is finite.
\end{proposition}
\begin{proof}
Suppose $f^{-1}(d)$ is infinite for some $d\in\Z$.
Consider some $x\in\N$.
Then $x-a$ divides $f(x)-d$ for every $a\in f^{-1}(d)$.
Thus, $f(x)-d$ has infinitely many divisors hence $f(x)=d$.
This shows that $f$ is the constant function with value $d$.
\end{proof}
\begin{corollary}\label{thm:varia e factorial}
Suppose $g:\N\to\N$ has integral difference ratios
and $a\in\Z\setminus\{0\}$;
 let $h_a:\N\to\Z$ be such that
$$
\text{if $a\neq1$ then}\ \
h_a(x)= \lfloor e^{1/a}\;a^{g(x)}\;  g(x)!\rfloor
\ ,\
h_1(x)= \left\{\begin{array}{ll}
1&\textit{if $g(x)=0$}\\
\lfloor e\;  g(x)!\rfloor & \textit{if $g(x)\neq0$}
\end{array}\right.
$$
1. If $g$ is not constant then $h_1$ differs from
$\lfloor e\;  g(x)!\rfloor$ on finitely many $x$'s.
\\
2. For $a\in\Z\setminus\{0\}$ the map $h_a$ has integral difference ratios.
\end{corollary}
\begin{proof}
1. By Proposition~\ref{p:bounded}, $g^{-1}(0)$ is finite.
Hence $h_1(x)$ and $\lfloor e\;  g(x)!\rfloor$
differ on finitely many $x$'s.
\\
2. Use Theorem~\ref{thm:e factorial}
and Proposition~\ref{p:closure composition}.
\end{proof}

\begin{example}
The following functions have integral difference ratios:
$$
x\mapsto\left\{\begin{array}{ll}
\left\lfloor e\;(x^2-5x+6)!)\right\rfloor & \text{if $x\neq2,3$}\\
1&\text{if $x=2,3$}
\end{array}\right.
\quad,\quad
x\mapsto\left\{\begin{array}{ll}
\lfloor e\; \lfloor e\;x! \rfloor !\rfloor & \text{if $x\neq0$}\\
2&\text{if $x=0$}
\end{array}\right.
$$
\end{example}

%%%%%%%%%%%%%%
\subsection{Examples with generalized hyperbolic functions}
\label{ss:e factorial generalized}
%%%%%%%%%%%%%%
%
The proof of Theorem~\ref{thm:e factorial} can be extended
to get more functions having integral difference ratios
and which are finite modifications of functions around the factorial functions.
Namely, for any given period $k\geq2$ and any $a\in\Z\setminus\{0\}$,
there exist  real numbers  $\alpha_0,\ldots,\alpha_{k-1}$ such that
the function $g$ defined by
$g(x)= \lfloor\alpha_s\,a^x\;x!\rfloor$  for $x\in s+k\N$
has integral difference ratios.
The main examples (Theorem~\ref{thm:e factorial}) correspond to
the (here excluded) degenerate case $k=1$.

First, we need simple results about the generalized hyperbolic functions
(a notion which goes back to V.~Ricatti, 1754,
for instance cf. \cite{Ungar1982,Muldoon2005}).
\begin{definition}\label{def:generalized hyperbolic}
Let $\gamma\in\RR$.
For\ $k,r\in\N$ such that $k\geq2$ and $r<k$,
the $\gamma$-hyperbolic function\ $F^\gamma_{k,r}:\RR\to\RR$\
is defined as follows: for $t\in\RR$,
$$
F^\gamma_{k,r}(t)\ =\ \sum\limits_{n\in k\N+r}
\gamma^{\lfloor n/k\rfloor}\,\dfrac{t^n}{n!}
\ =\ \sum\limits_{n\in \N} \gamma^n\,\dfrac{t^{kn+r}}{(kn+r)!}
$$
(so that
$F^1_{2,0}=\cosh$,\ $F^1_{2,1}=\sinh$,\
$F^{-1}_{2,0}=\cos$,\ $F^{-1}_{2,1}=\sin$\
are the usual hyperbolic and trigonometric functions).
\end{definition}
Recall some properties of the $1$-hyperbolic functions.
\begin{lemma}\label{l:parts of exp}
Let\ $k,r\in\N$ be such that $k\geq2$ and $r<k$.
\\
1. If $t\neq0$ and $-1<t<1$
then the sign of\ $F^1_{k,r}(t)$\ is that of\ $t^r$\ and
$$
0\ <\ \dfrac{F^1_{k,r}(t)}{t^r}\ \leq\ \cosh(|t|)
\ <\ \cosh(1)\ =\ 1.543\ldots\ .
$$
2. Let $\omega=e^{2i\pi/k}$\
be the canonical primitive $k$-th root of unity in the complex plane.
For all\ \ $t\in\RR$,
\begin{equation}\label{eq:phi k r}
F^1_{k,r}(t)\ =\ \dfrac{1}{k}\ 
\sum\limits_{\ell=0}^{\ell=k-1}\omega^{-\ell r}\; e^{\omega^\ell t}
\ =\ \dfrac{1}{k}\ 
\sum\limits_{\ell=0}^{\ell=k-1} e^{t\,\cos\left(\ell r\,\dfrac{2\pi}{k}\right)}
\;\cos\left(-\ell r\,\dfrac{2\pi}{k} + t\,
                 \sin\left(\ell\,\dfrac{2\pi}{k}\right)\right)
\end{equation}
3. For $q\geq 1$, the $q$-th derivative of $F^1_{k,r}$ is\
$(F^1_{k,r})^{(q)} = F^1_{k,s}$,
where\ $0\leq s <k$\ and\ $s\equiv r-q\pmod k$.
%\smallskip\\
%4. For any rational $t\neq0$ the number $F^1_{k,r}(t)$ is transcendental.
\end{lemma}
\begin{proof}
1. For $-1<t<1$, $t\neq0$, we have
\begin{equation}\label{eq:Fkr t over tr}
\dfrac{F^1_{k,r}(t)}{t^r}
\ =\ \sum_{n\in\N}\dfrac{t^{kn}}{(kn+r)!}
\ =\ 
\sum_{n\in\N}\left(\dfrac{|t|^{2nk}}{(2nk+r)!} 
    +\varepsilon\ \dfrac{|t|^{(2n+1)k}}{((2n+1)k+r)!}\right)
\end{equation}
where $\varepsilon=1$ if $t>0$ or $k$ is even
and $\varepsilon=-1$ if $t<0$ and $k$ is odd.

Since $k\geq2$ we have $((2n+1)k+r)!>(2nk+r)!$
and since $|t|<1$ and $t\neq0$ we have $|t|^{2nk}>|t|^{(2n+1)k}$.
In particular, for both possible values of $\varepsilon$,
the last sum in \eqref{eq:Fkr t over tr} consist of strictly positive terms
hence \ $F^1_{k,r}(t)/t^r$ is strictly positive.
Also, since $k\geq2$ and $r\geq0$,
\begin{multline*}
\dfrac{F^1_{k,r}(t)}{t^r}
\ \leq\ 
\sum_{n\in\N}\left(\dfrac{|t|^{2nk}}{(2nk+r)!} 
    + \dfrac{|t|^{(2n+1)k}}{((2n+1)k+r)!}\right)
\\\leq\ 
\sum_{n\in\N}\left(\dfrac{|t|^{4n}}{(4n)!} 
    + \dfrac{|t|^{4n+2}}{(4n+2)!}\right)
\ =\ \sum_{m\in\N}\dfrac{|t|^{2m}}{(2m)!}
\ =\ \cosh(|t|)
\ <\ \cosh(1)
\ =\ 1.543\ldots\ .
\end{multline*}
2. Arguing in the complex plane, for $\ell=0,\ldots,k-1$ and $t\in\RR$
and $r=0,\ldots,k-1$, we have
$$
\begin{array}{c}
e^{\omega^\ell t}
\ =\ \sum\limits_{n\in\N}\dfrac{\omega^{\ell n} t^n}{n!}
\ =\ \sum\limits_{u=0}^{u=k-1}
\sum\limits_{m\in\N} \omega^{\ell\,(km+u)}\,\dfrac{t^{km+u}}{(km+u)!}
\ =\ \sum\limits_{u=0}^{u=k-1}\omega^{\ell u}\,F^1_{k,u}(t)
\medskip\\
\dfrac{1}{k}\;
\sum\limits_{\ell=0}^{\ell=k-1} \omega^{-\ell r}\,e^{\omega^\ell t}
\ =\ 
\dfrac{1}{k}\;\sum\limits_{u=0}^{u=k-1}
\left(\sum\limits_{\ell=0}^{\ell=k-1} 
\omega^{\ell (u-r)}\right)\, F^1_{k,u}(t)
\ =\ F^1_{k,r}(t)
\end{array}
$$
since $\sum_{\ell=0}^{\ell=k-1}\omega^{\ell (u-r)}$
is equal to $k$ for $u=r$
and equal to $0$ for $u\in\{0,\ldots,k-1\}\setminus\{r\}$.
Now, since $t\in\RR$ so is $F^1_{k,r}(t)$
hence $F^1_{k,r}(t)$ is equal to the real part of the above expression
in the complex plane.
To conclude, observe that
\begin{multline*}
\omega^{-\ell r}\,e^{\omega^\ell t}
\ =\ e^{i\,\left(-\ell r\;2\pi/k\right)}\;
e^{t\,\cos\left(\ell \;2\pi/k\right)
  + i\,t\;\left(\sin\left(\ell \;2\pi/k\right)\right)}
\\
=\ e^{t\,\cos\left(\ell \;2\pi/k\right)}\;
e^{i\,\left(-\ell r\;2\pi/k 
                 + t\,\sin\left(\ell \;2\pi/k\right)\right)}\ .
\end{multline*}
3. Using the definition of $F^1_{k,r}$ as a series,
the derivative of $F^1_{k,r}$ is
$$
(F^1_{k,r})'(t)
= \sum\limits_{n\in k\N+r,\;n\geq1} \dfrac{t^{n-1}}{(n-1)!}
=\left\{\begin{array}{ll}
F^1_{k,r-1}(t)&{\text \ if \ } r\geq 1\smallskip\\
F^1_{k,k-1}(t)&{\text \ if \ } r=0 
\end{array}\right.
$$
An obvious induction on $q$ concludes the proof.
%\smallskip\\
%4. The case $k=2$ is the classical result that
%$\cosh(t)$ and $\sinh(t)$
%are transcendental when $t\neq0$ is algebraic.
%%(for instance, cf. \cite{Niven1956} Theorem 2.9 page 22).
%For the general case, we use a powerful tool in number theory:
%%
%\begin{quote}{\bf Lindemann-Weierstrass' theorem (1885).}
%\it Let $\alpha_1,\ldots,\alpha_n$ be distinct algebraic numbers, and let
%$\beta_1,\ldots,\beta_n$ be non-zero algebraic numbers.
%Then $\beta_1\;e^{\alpha_1}+\cdots+\beta_n\;e^{\alpha_n}\neq0$.
%\end{quote}
%\change{formulation}
%By point 2, we know that
%$F^1_{k,r}(t)\ =\ \dfrac{1}{k}\ 
%\sum\limits_{\ell=0}^{\ell=k-1}\omega^{-\ell r}\; e^{\omega^\ell t}$.
%Being a $k$-th root of $1$, the number $\omega$ is algebraic
%and so are its powers.
%For  any  non null polynomial with integral coefficients  $P(X)\in\Z[X]$, 
% $P(F^1_{k,r}(t))$ can be written as a sum like the one
%in Lindemann-Weierstrass' theorem hence is non null.
%Thus, the real number $F^1_{k,r}(t)$ is transcendental.
\end{proof}

\begin{theorem}\label{thm:e factorial generalized}
For any\ $a\in\Z\setminus\{0\}$,
$k\in\N\setminus\{0,1\}$, $r\in\{0,\ldots,k-1\}$,
let $\+F_{a,k,r}$ and $\+C_{a,k,r}$ be the following functions $\N\to\Z$~:
$$
\begin{array}{clc}
\+F_{a,k,r}(x)&\quad&\+C_{a,k,r}
\\
\left\{\begin{array}{l}
\left\lfloor F^1_{k,0}(1/a)\;a^x\;x!\right\rfloor
\\\qquad\vdots\\
\left\lfloor F^1_{k,k-1}(1/a)\;a^x\;x!\right\rfloor
\end{array}\right.
&
\begin{array}{l}
\text{if } x\in k\N+r
\\\qquad\vdots\\
\text{if } x\in k\N+r+k-1
\end{array}
&
\left.\begin{array}{l}
\left\lceil F^1_{k,0}(1/a)\;a^x\;x!\right\rceil
\\\qquad\vdots\\
\left\lceil F^1_{k,k-1}(1/a)\;a^x\;x!\right\rceil
\end{array}\right\}
\end{array}
$$
Let us denote $f\oplus\{(0,n_0),\ldots,(\ell,n_\ell)\}$
the function $g$ such that
$g(x)=f(x)$ if $x>\ell$ and $g(t)=n_t$ if $0\leq t\leq \ell$.

The following functions $\N\to\Z$ have integral difference ratios:
\begin{conditionsbullet}
\item
Case $|a|\geq2$ and $r=0$.
$\+F_{a,k,0}$ and $\+C_{a,k,0}$,
\item
Case $1\leq r<k$ and
either $a\geq2$ or $a\leq-2$ and $k-r$ is even.\\
$\+F_{a,k,r}\oplus\{(0,0),\ldots(r-1,0)\}$ 
and $\+C_{a,k,r}\oplus\{(0,1),\ldots(r-1,1)\}$,
\item
Case $1\leq r<k$ and $a\leq-2$ and $k-r$ is odd.\\
$\+F_{a,k,r}\oplus\{(0,-1),\ldots(r-1,-1)\}$ 
and $\+C_{a,k,r}\oplus\{(0,0),\ldots(r-1,0)\}$,
\item
Case $a=1$ and $r=0$.
$\+F_{1,k,0}\oplus\{(0,1)\}$ and $\+C_{1,k,0}\oplus\{(0,2)\}$,
\item
Case $a=1$ and $1\leq r<k$.\\
$\+F_{1,k,r}\oplus\{(0,0)\ldots(r-1,0)\}$ 
and $\+C_{a,k,r}\oplus\{(0,1)\ldots(r-1,1)\}$,
\item
Case $a=-1$ and $r=0$ and $k$ is even.\\
$\+F_{-1,k,0}\oplus\{(0,0)\}$ and $\+C_{-1,k,0}\oplus\{(0,1)\}$,
\item
Case $a=-1$ and $r=0$ and $k$ is odd.\\
$\+F_{-1,k,0}\oplus\{(0,1)\}$ and $\+C_{-1,k,0}\oplus\{(0,0)\}$,
\item
Case $a=-1$ and $1\leq r<k$ and $k$ is even.\\
$\+F_{-1,k,r}\oplus\{0,0)\ldots(r-1,0\}$ 
and $\+C_{a,k,r}\oplus\{(0,1)\ldots(r-1,1)\}$.
\item
Case $a=-1$ and $1\leq r<k$ and $k$ is odd.\\
$\+F_{-1,k,r}\oplus\{(0,-1)\ldots(r-1,-1)\}$ 
and $\+C_{a,k,r}\oplus\{(0,0)\ldots(r-1,0)\}$.
\end{conditionsbullet}
\end{theorem}
We first give an example, then we will prove the Theorem.
\begin{example}
The functions corresponding to $a=k=2$ are
$$
\begin{array}{c}
\fbox{\text{Case\ }r=0}\\
\left\{\begin{array}{ll}
\lfloor\cosh(1/2)\;2^x\;x!\rfloor&\text{if\quad}x\in2\N\\
\lfloor\sinh(1/2)\;2^x\;x!\rfloor&\text{if\quad}x\in2\N+1
\end{array}\right.
\end{array}
\ \
\begin{array}{c}
\fbox{\text{Case\ }r=1}\\
\left\{\begin{array}{ll}
0&\text{if\quad}x=0\\
\lfloor\sinh(1/2)\;2^x\;x!\rfloor&\text{if\quad}x\in2\N+1\\
\lfloor\cosh(1/2)\;2^x\;x!\rfloor&\text{if\quad}x\in2\N+2
\end{array}\right.
\end{array}
$$
The coefficients $F^1_{3,r}(1/3)$ occurring in the functions
corresponding to $a=k=3$ are given by the following formulas
$$
\begin{array}{rcl}
F^1_{3,0}(t)&=&(1/3)\;\left(e^t+2\,e^{-t/2}\,
\cos\left(t\,\sqrt{3}/2\right)\right)\\
F^1_{3,1}(t)&=&(1/3)\;\left(e^t-2\,e^{-t/2}\,
\cos\left(t\,\sqrt{3}/2 + \pi/3\right)\right)\\
F^1_{3,2}(t)&=&(1/3)\;\left(e^t-2\,e^{-t/2}\,
\cos\left(t\,\sqrt{3}/2 - \pi/3\right)\right)
\end{array}
$$
\end{example}
\begin{proof}[Proof of Theorem~\ref{thm:e factorial generalized}]
Since $F^1_{k,s}(t)=\sum_{n\in \N}t^{kn+s}/(kn+s)!$
For $s\in\{0,\ldots,k-1\}$, we have
\begin{eqnarray*}
F^1_{k,s}(0)&=&\sum_{n\in \N}0^{kn+s}/(kn+s)!
\ =\ \left\{\begin{array}{ll}
1&\text{if $s=0$}\\
0&\text{if $1\leq s<k$}
\end{array}\right.
\end{eqnarray*}

Since the $q$-th derivative of $F^1_{k,s}$ is $F^1_{k,s'}$
with $0\leq s'<k$ and $s'\equiv s-q\pmod k$
(cf. Lemma~\ref{l:parts of exp}),
we have
$$
\left\{\begin{array}{rcll}
(F^1_{k,s})^{(q)}(0)&=&1&\text{if $q\in k\N+s$}
\\
(F^1_{k,s})^{(q)}(0)&=&0&\text{otherwise}
\end{array}\right.
$$
and $(F^1_{k,s})^{(k(u+1)+s)}=F^1_{k,0}$.
Thus, the Taylor-Lagrange development at order $k(u+1)+s-1$,
of $F^1_{k,s}$ at $t$ is, for some $\theta\in]0,1[$,
\begin{eqnarray}\notag
F^1_{k,s}(t)&=&
\left(\sum\limits_{q=0}^{q=k(u+1)+s-1}
\dfrac{t^q}{q!}\;(F^1_{k,s})^{(q)}(0)\,\right)
+\dfrac{t^{k(u+1)+s}}{(k(u+1)+s)!}\;F^1_{k,0}(\theta\,t)
\\\label{eq:init of varphi k r}
&=&
\left(\sum\limits_{m=0}^{m=u}\dfrac{t^{km+s}}{(km+s)!}\right)
+\dfrac{t^{k(u+1)+s}}{(k(u+1)+s)!}\;F^1_{k,0}(\theta\,t)
\ .
\end{eqnarray}
For $k,r\in\N$ such that $k\geq2$ and $0\leq r<k$,
let $f_{a,k,r}:\N\to\Z$ be the function associated to the Newton series
\begin{equation}\label{eq:fakr}
f_{a,k,r}(x) \ =\ \sum_{n\in k\N+r}a^n\; n!\;\dbinom{x}{n}\ .
\end{equation}
By Corollary \ref{rk:main1}, $f_{a,k,r}$ has integral difference ratios.
Recall that $\dbinom{x}{n}=0$ for $n>x$.
Thus,
\begin{equation}\label{eq:fak for less than r}
f_{a,k,r}(x)=0\qquad\text{if\ \ $0\leq x<r$}
\end{equation}
Also, for $u\in\N$,  $s \in\{0, . . . , k-1\}$ and $x=uk+r+s$, we have 
$$
\begin{array}{rclll}
f_{a,k,r}(x)&=& \sum\limits_{p=0}^{p=u}
a^{pk+r}\,(pk+r)!\;\dbinom{uk+r+s}{pk+r}
&=& a^x\, x!\;\sum\limits_{p=0}^{p=u}\dfrac{(1/a)^{x-pk-r}}{(x-pk-r)!}
\smallskip\\
&=&a^x\, x!\;\sum\limits_{p=0}^{p=u}
\dfrac{(1/a)^{k(u-p)+s}}{(k(u-p)+s)!}
&=&a^x\, x!\;\sum\limits_{m=0}^{m=u}
\dfrac{(1/a)^{km+s}}{(km+s)!}
\end{array}
$$
Using equation \eqref{eq:init of varphi k r} with $t=1/a$, we get,
for $x=uk+r+s$,
\begin{eqnarray}\notag
f_{a,k,r}(x)&=&a^x\,x!\; F^1_{k,s}(1/a) - 
a^x\,x!\;\dfrac{(1/a)^{k(u+1)+s}}{(k(u+1)+s)!}\;
                   F^1_{k,0}(\theta/a)
\\\label{eq:def Delta}
\text{and letting}\quad \Delta&=&a^x\,x!\; F^1_{k,s}(1/a) - f_{a,k,r}(x)
\\\notag
\text{we have}\quad 
\Delta&=&a^x\,x!\;\dfrac{(1/a)^{k(u+1)+s}}{(k(u+1)+s)!}\;
                     F^1_{k,0}(\theta/a)
\\\label{eq:Delta}
&=&\dfrac{F^1_{k,0}(\theta/a)}
               {a^{k-r}\prod_{j=r+1}^{j=k}ku+s+j}
\end{eqnarray}

Since $a\in\Z\setminus\{0\}$, we have $0<|\theta/a|<1$.
Also, $(\theta/a)^0=1$
and point 1 of Lemma~\ref{l:parts of exp} yields
\begin{equation}\label{eq:cosh1}
0\ <\ F^1_{k,0}(\theta/a)\ <\ \cosh(1)\ =\ 1.543\ldots\ .
\end{equation}
Since $x=ku+r+s$, we have $ku+s+j=x+1$ for $j=r+1$.
Inequalities~\eqref{eq:cosh1} and $k-r\geq1$ insure that
\begin{equation*}%\label{eq:R}
0\ <\ |\Delta|\ <\ \dfrac{1.543\ldots}{|a|^{k-r}\,(x+1)}
\ <\ \dfrac{1.543\ldots}{|a|\,(x+1)}\ .
\end{equation*}
Equation~\eqref{eq:Delta} shows that
the sign of $\Delta$ is that of $a^{k-r}$.

Since equation~\eqref{eq:fak for less than r} gives $f_{a,k,r}(x)$
for $0\leq x<r$, it suffices to consider the values $x\geq r$.

{\it Case $a\geq2$ and Case $a\leq-2$ and $k-r$ even.}
For every $x\in\N$ we have $0<\Delta< 1$.
Since $f_{a,k,r}(x)\in\Z$,
the definition of $\Delta$ given by \eqref{eq:def Delta} yields
\begin{equation}\label{eq:a greater 2}
\text{for $x\in k\N+r+s$}\quad
\left\{\begin{array}{rcl}
f_{a,k,r}(x)&=&\left\lfloor a^x\,x!\; F^1_{k,s}(1/a) \right\rfloor
\medskip\\
f_{a,k,r}(x)+1&=&\left\lceil a^x\,x!\; F^1_{k,s}(1/a) \right\rceil
\end{array}\right.
\end{equation}

{\it Case $a\leq-2$ and $k-r$ odd.}
For every $x\in\N$ we have $-1<\Delta< 0$ hence
\begin{equation}\label{eq:a less -2}
\text{for $x\in k\N+r+s$}\quad
\left\{\begin{array}{rcl}
f_{a,k,r}(x)&=&\left\lceil a^x\,x!\; F^1_{k,s}(1/a) \right\rceil
\medskip\\
f_{a,k,r}(x)-1&=&\left\lfloor a^x\,x!\; F^1_{k,s}(1/a) \right\rfloor
\end{array}\right.
\end{equation}

{\it Case $a=1$ and Case $a=-1$ and $k-r$ even.}
Then $0<\Delta< 1$ for all $x\geq1$ hence
\eqref{eq:a greater 2} holds with the extra hypothesis $x\geq1$.

{\it Case $a=-1$ and $k-r$ odd.}
Then $0<\Delta< 1$ for all $x\geq1$ hence
\eqref{eq:a less -2} holds with the extra hypothesis $x\geq1$.

In both cases $a=1$ and $a=-1$, we also have
\begin{eqnarray*}
f_{-1,k,r}(0)\ =\ f_{1,k,r}(0)&=&\left\{\begin{array}{ll}
1&\text{if\ $r=0$}\\
0&\text{if\ $1\leq r<k$}
\end{array}\right.
\end{eqnarray*}
Thus, the functions mentioned in the theorem
are among the functions
$f_{a,k,r}$, $f_{a,k,r}-1$ and $f_{a,k,r}+1$,
all of which have integral difference ratios.
\end{proof}

\noindent{\bf Open problem.}
Theorems~\ref{thm:e factorial} and \ref{thm:e factorial generalized},
give simple analytic expressions for the functions $\N\to\Z$
associated to Newton series
$$
\sum_{n\in \N} a^n\,n!\,\dbinom{x}{n}
\quad,\quad
\sum_{n\in k\N+r} a^n\,n!\,\dbinom{x}{n}
$$
for $a\in\Z\setminus\{0\}$.
Theorem~\ref{thm:main1} invites to look at other natural Newton series
such as
$$
\sum_{n\in \N} \lcm(n)\,\dbinom{x}{n}
\quad,\quad
\sum_{n\in k\N+r} \lcm(n)\,\dbinom{x}{n}\ .
$$
Is it possible to give analytic expressions
to the associated functions?

%
%%%%%%%%%%%%%%%%%%%%%%%%%%%%
\subsection{Asymptotic equivalence}\label{ss:asymptotic}
%%%%%%%%%%%%%%%%%%%%%%%%%%%%
%
A simple consequence of Theorem~\ref{thm:main1} insures that
{\em any} function which grows fast enough
is asymptotically equivalent to
a function having integral difference ratios.
\begin{theorem}\label{thm:equivalent to idr}
For every function $f:\N\to\Z$ there exists some function
$g:\N\to\Z$ which has integral difference ratios and such that,
for all $x\in\N$,
\begin{equation*}
0\leq f(x)-g(x) \leq 2^x\,\lcm(x)\ .
\end{equation*}
In particular, if there is some $\varepsilon>0$ such that
$|f(x)|\geq(2e+\varepsilon)^x$ for all $x$ large enough
then $\lim_{x\to+\infty}\dfrac{f(x)}{g(x)}=1$,
i.e. $f$ and $g$ are asymptotically equivalent. 
\end{theorem}
\begin{proof}
Consider the Newton coefficients $(a_k)_{k\in\N}$ of $f$
(cf. Proposition~\ref{p:Newton}).
Let $a_k=\lcm(k)q_k+b_k$ where $0\leq b_k<\lcm(k)$
and set $\widetilde{a_k}=\lcm(k)q_k$.
Since $\lcm(k)$ divides $\widetilde{a_k}$ for all $k$'s,
the function $g(x)=\sum_{k\in\N}\widetilde{a_k}\dbinom{x}{k}$
has integral difference ratios.
Also, $0\leq f(x)-g(x) = \sum_{k\leq x}b_k\dbinom{x}{k}
<\sum_{k\leq x}\lcm(k)\dbinom{x}{k} 
<\left(\sum_{k\leq x}\dbinom{x}{k}\right)\,\lcm(x)
= 2^x\,\lcm(x)$
so that
$0\leq\left| 1-\dfrac{g(x)}{f(x)}\right|
<\dfrac{2^x\,\lcm(x)}{|f(x)|}$
which tends to $0$ when $x$ tends to $+\infty$
thanks to majoration \eqref{eq:lcm exp} in Remark~\ref{rk:lcm}
and the assumption on $f$.
\end{proof}
\begin{corollary}\label{cor:alpha factorial equivalent}
For any real numbers $\alpha, a$ such that $a>0$,
the functions $x\mapsto\lfloor\alpha\; a^x\;x!\rfloor$ 
and  $x\mapsto\lceil\alpha\; a^x\;x!\rceil$  from $\N$ to $\Z$ are asymptotically equivalent to
some function $g:\N\to\Z$ having integral difference ratios.
\end{corollary}
\begin{proof}
Stirling's formula insures that $a^x\;x!$ satisfies
the growth condition of Theorem~\ref{thm:equivalent to idr}.
\end{proof}
\begin{remark}
Theorem~\ref{thm:equivalent to idr} shows that there
are functions having integral difference ratios
which grow arbitrarily fast.
In particular, functions growing much faster than the $a^x\;x!$
with $a\in\Z$.
\end{remark}

%
%%%%%%%%%%%%%%%%%%%%%%%%%%%%
%%%%%%%%%%%%%%%%%%%%%%%%%%%%
%%%%%%%%%%%%%%%%%%%%%%%%%%%%
\section{Outside the family of functions with integral difference ratios}
\label{s:out}
%%%%%%%%%%%%%%%%%%%%%%%%%%%%
%%%%%%%%%%%%%%%%%%%%%%%%%%%%
%%%%%%%%%%%%%%%%%%%%%%%%%%%%
%
As expected, the examples stated in Theorems~\ref{thm:e factorial}
and \ref{thm:e factorial generalized} are kind of exceptions:
most functions similar to these examples {\it do not}
have rational difference ratios. 
Worse, though they are asymptotically equivalent to functions
having rational difference ratios
(cf. Theorem~\ref{thm:equivalent to idr})
they can not be uniformly approximated by such functions.

In fact, it turns out that proving non uniform closeness
is a very manageable approach to prove failure of the
integral difference ratios property.

Using a classical result
in the theory of distribution modulo one,
we give a general negative result for uniform closeness
involving a measure zero set of possible exceptions.
Then we look at the problem for some particular classes of functions.

%%%%%%%%%%%%%%%%%%%%%%%%%%%%
\subsection{Uniform closeness}
%%%%%%%%%%%%%%%%%%%%%%%%%%%%
%
\begin{definition}
Two functions $f,g:\N\to\RR$ are \emph{uniformly close} if $f-g$ is bounded,
i.e. there exists $M$ such that $|f(x)-g(x)|\leq M$ for all $x\in\N$.
\end{definition}
Some straightforward closure properties will be
used together with Propositions~\ref{p:closure sum}
and \ref{p:closure composition}.
\begin{proposition}\label{p:closure sum close}
If $\varphi,\psi:\N\to\Z$ are uniformly close to $f,g:\N\to\Z$
then $\varphi+\psi$ (resp. $k\varphi$) is uniformly close to
$f+g$ (resp. $kf$) for all $k\in\Z$.
\end{proposition}
\begin{proposition}\label{p:closure composition close}
Let $f,g, \varphi, \psi$ be functions $\N\to\RR$.
If $\varphi$ is uniformly close to $f$,
and $\psi$ differs from $g$ on finitely many points,
then $\varphi\circ\psi$ is uniformly close to $f\circ g$.
\end{proposition}
\begin{proof}
Suppose $|\varphi(t)-f(t)|\leq M$ for all $t\in\N$
and $a_1,\ldots,a_p$ are the points on which $\psi$ differ from $g$.
Let $L=\max(M,\max_{i=1,\ldots,p}|\varphi(\psi(a_i))-f(g(a_i))|)$.
Then $|\varphi(\psi(a_i))-f(g(a_i))|\leq L$ and, for $x\neq a_1,\ldots,a_p$,
$|\varphi(\psi(x))-f(g(x))|=|\varphi(g(x))-f(g(x))|\leq L$.
\end{proof}

%%%%%%%%%%%%%%%%%%%%%%%%%%%%
\subsection{A general negative result for uniform closeness to
functions having integral difference ratios}
\label{ss:non idr close}
%%%%%%%%%%%%%%%%%%%%%%%%%%%%
%
Recall the following classical notion.
\begin{definition}\label{def:fractional part}
For $A\in\N\setminus\{0\}$ and $t\in\RR$,
the {\em$A$-fractional part} of $t$ is
$\{t\}_A=t-A\,\lfloor t/A\rfloor$,
i.e. $\{kA+u\}_A=u$ for any $k\in\Z$ and $u\in[0,A[$.
The 1-fractional part   $\{t\}_1$ is simply denoted by $\{t\}$.
\end{definition}
Before entering the wanted negative result,
we first observe the following fact.
\begin{lemma}\label{l:uniformly close}
Suppose $\varphi:\N\to\RR$ is uniformly close to some function $f:\N\to\Z$
such that $n$ divides $f(n)-f(0)$ for all $n\geq1$.
Then, for all $A\in\N$ big enough,
the sequence of $A$-fractional parts $(\{\varphi(nA)\}_A)_{n\in\N}$
is not dense in $[0,A]$. 
\end{lemma}
\begin{proof}
Let $g(n)=f(n)-f(0)$.
Then $\varphi$ is also uniformly close to $g$ and
$n$ divides $g(n)$ for all $n$.
Let $M>0$ be such that $|\varphi(n)-g(n)|\leq M$ for all $n\in\N$.
Consider any $A\in\N$ such that $A>2M$.
Then $\varphi(nA)\in[g(nA)-M,g(nA)+M]$.
Since $nA$ (hence $A$) divides $g(nA)$, we have
$\{g(nA)+u\}_A=u$ for $u\in[0,A[$
and $\{g(nA)-v\}_A=\{(g(nA)-A+(A-v)\}_A=A-v$ for $-v\in]-A,0[$.
In particular,
\begin{eqnarray*}
\{\{y\}_A \mid y\in[g(nA),g(nA)+M]&=&[0,M]
\\
\{\{y\}_A \mid y\in[g(nA)-M,g(nA)[\}&=&[A-M,A[
\end{eqnarray*}
so that $\{\varphi(nA)\}_A\in[0,M]\cup[A-M,A[$ for all $n\in\N$.
Since $A>2M$ we have $M<A-M$
hence the non empty open subinterval $]M,A-M[$ of $[0,A[$
contains no point of the sequence $(\{\varphi(nA)\}_A)_{n\in\N}$.
\end{proof}
We shall use a result
from the theory of uniform distribution modulo one.
\begin{definition}\label{def:ud}
A sequence $(t_n)_{n\in\N}$ is uniformly distributed
modulo one if, for all $0\leq a<b\leq1$,
the proportion of $i$'s in $\{0,\ldots,n-1\}$ such that
the 1-fractional part $\{t_i\}$ is in $[a,b[$
tends to $b-a$ when $n$ tends to $+\infty$,
i.e. $\lim_{n\to+\infty}\dfrac{1}{n}\,
\card\{i\in\{0,\ldots,n-1\}\mid \{t_i\}\in[a,b[\}=b-a$.
\end{definition}
\begin{theorem}[Koksma, 1935, cf. Corollary 4.3 in \cite{KN}]
\label{thm:koksma}
Let $(\lambda_n)_{n\in\N}$ be a sequence of reals
such that $\inf\{|\lambda_m-\lambda_n| \mid m\neq n\} >0$.
Then, for almost all real numbers $\alpha$,
the sequence $(\alpha\,\lambda_n)_{n\in\N}$ is uniformly distributed
modulo one.
\end{theorem}
\begin{remark}\label{rk:ae}
The ``almost everywhere" restriction cannot be removed in
Theorem~\ref{thm:koksma}.
It is known (cf. Example~4.2 in \cite{KN})
that if $t$ is a Pisot-Vijayaraghavan number
(in particular, if $t$ is a rational number or $t$ is the golden number
$(1+\sqrt{5})/2$) then the sequence $(t^n)_{n\in\N}$ has no limit point
except possibly $0$ or $1$ hence is not uniformly distributed modulo one.
\end{remark}
We now come to the wanted general negative result.
\begin{theorem}[Almost everywhere negative result]\label{thm:ae}
Let $(\lambda_n)_{n\in\N}$ be a sequence of real numbers
satisfying the condition
$\inf\{|\lambda_m-\lambda_n| \mid m\neq n\} >0$.
Then, for almost all real numbers $\alpha$,
the function $n\mapsto\alpha\,\lambda_n$ is not uniformly close to
any function having integral difference ratios.
\end{theorem}
\begin{proof}
The assumed hypothesis on the $\lambda_n$'s, insures that
we can apply Koksma's theorem for each sequence
$(\lambda_{nA}/A)_{n\in\N}$
with $A\in\N\setminus\{0\}$.
Since sets of measure zero are closed under countable union,
Koksma's theorem insures that there exists a set $X\subseteq\RR$
such that $\RR\setminus X$ has measure zero
and, for all $\alpha\in X$, all the sequences
$(\alpha\,\lambda_{nA}/A)_{n\in\N}$,
with $A\in\N\setminus\{0\}$, are uniformly distributed modulo one
hence are dense in $[0,1[$.
Applying the homothety $t\mapsto At$, we see that,
for all $\alpha\in X$,
\begin{equation}\label{eq:dense}
\textit{All the sequences $(\alpha\,\lambda_{nA})_{n\in\N}$,
with $A\in\N\setminus\{0\}$, are dense in $[0,A[$.}
\end{equation}

By way of contradiction, suppose that, for some $x\in X$,
the function $n\mapsto\alpha\,\lambda_n$ is uniformly close to
some function $f:\N\to\N$ having integral difference ratios.
Observe that $n$ divides $f(n)-f(0)$ for all $n\in\N$,
so that we can apply Lemma~\ref{l:uniformly close}:
for $A\in\N$ large enough,
the sequence $(\alpha\,\lambda_{nA})_{n\in\N}$ is not dense in $[0,A[$.
This contradicts property~\eqref{eq:dense}. 
\end{proof}
Theorem~\ref{thm:ae} shows that the constant $e^{1/a}$
has a crucial role in Theorem~\ref{thm:e factorial}.
\begin{corollary}\label{cor:ae}
1. Let $f\colon \N\to\RR$ be a function such that
$\inf_{k\leq m<n}|f(m)-f(n)| >0$ for some $k$.
For almost every $\alpha\in\RR$,
the real-valued function $n\mapsto \alpha f(n)$ is not uniformly close
to a function having integral difference ratios.
\\
2. Let $a\in\RR\setminus\{0\}$.
For almost every $\alpha\in\RR$, the $\Z$-valued functions
$n\mapsto\lfloor \alpha\ a^n\ n!\rfloor$ and
$n\mapsto\lceil \alpha\ a^n\ n!\rceil$
are not uniformly close to functions having integral difference ratios.
\end{corollary}
\begin{proof}
1. First modify $f$ on $\{0,\ldots,k\}$ to obtain $f'$ such that
$\inf_{m<n} |f'(m)-f'(n)|  >0$.
Apply then Theorem~\ref{thm:ae} with $\lambda_n=f'(n)$, thus $\alpha f'$ is not uniformly close to
any function having integral difference ratios. The same holds for $\alpha f$ which is equal to 
$\alpha f'$ except on a finite number of points (uniform closeness is not modified by finitely many changes).

2.
Applying  (1)  with $f\colon n\mapsto a^n\;n!$ we see that 
 $\alpha f$ is not uniformly close to any $\N\to\Z$ function having
integral difference ratios; the same holds for  the $\N\to\Z$ functions
$\lfloor\alpha\,a^n\;n!\rfloor$ and
$\lceil\alpha\,a^n\;n!\rceil$ which are
 uniformly close to $\alpha f$.
\end{proof}

\begin{remark}\label{rk:a in R}
1.  Note the difference with Corollary \ref{cor:alpha factorial equivalent}.
\\
2. In Theorem~\ref{thm:e factorial}, the parameter $a$ is taken in $\Z$
so that the Newton series
$\sum_{n\in\N}a^n\,n!\,\binom{x}{n}$
takes its values in $\Z$.
In the above corollary, we can take $a$ in $\RR$.
\end{remark}
\begin{example}\label{ex:ae}   Let $P$ be a non constant polynomial with real coefficients; for almost every $\alpha\in\RR$, the  function $\alpha P$
is not uniformly close to any function having integral difference ratios.
\end{example}

The analog result for the particular constants
$F^1_{k,0}(1/a)$,\ldots,$F^1_{k,k-1}(1/a)$
in Theorem~\ref{thm:e factorial generalized}
requires an easy extension of Theorem~\ref{thm:ae}.
\begin{theorem}\label{thm:ae bis}
Let $k,r,s\in\N$ be such that $k\geq1$ and $0\leq r<k$.
Let $(\lambda_n)_{n\in k\N+r+s}$ be a sequence of real numbers
satisfying the condition
$\inf\{|\lambda_m-\lambda_n| \mid 
                     m,n\in k\N+r+s,\ m\neq n\} >0$.
Then, for almost all $\alpha\in\RR$,
no function $f:\N\to\RR$ such that
$f(kx+r+s)=\alpha\lambda_x$ for all $x\in\N$
can be uniformly close to
some function having integral difference ratios.
\end{theorem}
\begin{proof}
Arguing as in the proof of Theorem~\ref{thm:ae},
first, we see that, for all $\alpha\in X$,
\begin{equation*}
\textit{All the sequences
$(\alpha\,\lambda_{nkA+r+s})_{n\in\N}$,
with $A\in\N\setminus\{0\}$, are dense in $[0,A[$}
\end{equation*}
and then we conclude using Lemma~\ref{l:uniformly close}.
\end{proof}
\begin{corollary}\label{cor:ae bis}
Let $a\in\RR\setminus\{0\}$ and $k,s\in\N$
be such that $k\geq2$.
For almost all $\alpha\in\RR$,
every function $f:\N\to\Z$  such that
$f(x)=\lfloor \alpha\ a^x \ x!\rfloor$
for all $x\in k\N+s$
has non integral difference ratios.
Idem with $\lceil\ldots\rceil$ in place of $\lfloor\ldots\rfloor$.
\end{corollary}

%%%%%%%%%%%%%%%%%%%%%%%%%%%%
\subsection{Non integral polynomial functions}
\label{ss:poly not idr}
%%%%%%%%%%%%%%%%%%%%%%%%%%%%
%
Apart the obvious fact (cf. Corollary~\ref{cor:poly})
that polynomials with coefficients in $\Z$ 
have integral difference ratios,
there are only negative results for polynomials with real coefficients.
In particular, the positive result with
$\lfloor e^{1/a}\,a^x\;x!\rfloor$
(cf. Theorem~\ref{thm:e factorial})
has no analog with polynomials. Theorem \ref{thm:uniformly close to poly}
completes Corollary~\ref{cor:poly}, Example~\ref{ex:ae} and Corollary~\ref{cor:ae}.
\begin{theorem}\label{thm:uniformly close to poly}
Let $P(x)=\sum_{i=0}^{i=k}\alpha_i x^i$,   %$\alpha_i\in \RR$ for $i=1,\ldots,k$
be a  polynomial with real coefficients 
and consider it as a function $\N\to\RR$.
The following conditions are equivalent.
\begin{conditionsiii}
\item
The coefficients $\alpha_i$, $i=1,\ldots,k$ of $P$ are in $\Z$.
\item
$P$ maps $\N$ into $\Z$ and has integral difference ratios.
\item
$P:\N\to\RR$ is uniformly close to some function $\N\to\Z$
having integral difference ratios.
\end{conditionsiii}
\end{theorem}
%
%\begin{remark}
%The above statement subsumes the obvious direct application of
%Theorem~\ref{thm:ae}.
%\end{remark}
%
\begin{proof}
$(i)\Rightarrow(ii)$ is Corollary~\ref{cor:poly}
and $(ii)\Rightarrow(iii)$ is trivial.
We prove $(iii)\Rightarrow(i)$.
Suppose condition $(iii)$ is true; let $N_{a,b}, K, \theta_x$ be defined as in
points (a) and (b) below.
\begin{conditionsabc}
%\item
%$P(x)=\sum_{i=0}^{i=k}\alpha_i x^i$ where $k\geq1$ is the degree of $P$,
%the $\alpha_i$'s are in $\RR$ and $\alpha_k\neq0$.
\item
$\varphi:\N\to\Z$ has integral difference ratios.
For $a,b\in\N$, $a\neq b$, let $N_{a,b}\in\Z$ be such that
$\varphi(a)-\varphi(b)=N_{a,b}(a-b)$,
\item
$P$ is uniformly close to $\varphi$ and $K, \theta_x$ are such that,
for all $x\in\N$, $|P(x)-\varphi(x)|\leq K$,
i.e. $P(x)=\varphi(x)+\theta_xK$ for some $\theta_x\in[-1,1]$.
\end{conditionsabc}
By induction on the degree of $P$ we prove that all coefficients of $P$, except may be $\alpha_0$
are in $\Z$. 

{\em Basis: } if $P=\alpha_1 x +\alpha_0$ has degree one\, then $\alpha_1\in\Z$. 
By point (a), $\varphi(x)-\varphi(0)=N_{x}x$, $N_x\in\N$.
By point (b), $P(x)-P(0)=\varphi(x)-\varphi(0)+ (\theta_x-\theta_0) K$, $\theta_x,\theta_0\in[-1,1]$.
Let $\alpha_1= \lfloor\alpha_1\rfloor+\theta$, $\theta\in[-0,1[$.
Thus for all $x$, $P(x)-P(0)= \lfloor\alpha_1\rfloor x+\theta x=N_{x}x+(\theta_x-\theta_0) K$,
hence $x( \lfloor\alpha_1\rfloor-N_{x}+\theta)=(\theta_x-\theta_0) K$.
Assume by contradiction $\alpha_1\not\in\Z$ and $\theta\not= 0$,  and let $x>\dfrac{3K}{\min(\theta,1-\theta)}$: noting that  $ \lfloor\alpha_1\rfloor-N_{x}\in\Z$, we have
$x( \lfloor\alpha_1\rfloor-N_{x}+\theta)>3K$, contradicting $(\theta_x-\theta_0) K<2K$.
This show that, if $P$ has  degree one, then  $\alpha_1\in\Z$.

{\em Induction: } it suffices to prove that $\alpha_k$,  the leading coefficient of $P$,
is in $\Z$.
Then, an  induction on the degree of $P$ concludes the proof:
if $\alpha_k\in\Z$ then $P(x)-\alpha_k x^k$ also satisfies condition $(iii)$
(by Proposition~\ref{p:closure sum close}) and has degree $k-1$.

An easy way to single out $\alpha_k$ is to consider the $k$-th derivative
of $P$ which is $k!\;\alpha_k$.
Since we are in a discrete context with functions defined on $\N$ and not on $\RR$,
we turn to finite differences.

For $1\leq n\leq k$,
define a polynomial $P^{(n)}$ by the following induction:
$$
\begin{array}{rcl}
P^{(1)}(x)&=&\dfrac{P(2x) - P(x)}{x}
\ \ =\ \ \sum_{i=1}^{i=k}\alpha_i (2^i-1) x^{i-1}
\medskip\\
\textit{if $1<n\leq k$}\quad
P^{(n)}(x)&=&\dfrac{P^{(n-1)}(2x)-P^{(n-1)}(x)}{x}
\\&=&\sum_{i=n}^{i=k}\alpha_i \left(\prod_{j=i-n+1}^{j=i}(2^j-1)\right) x^{i-n}
\end{array}
$$

\medskip

\noindent{\bf Claim 1.}{ \it
For every $n\in\{1,\ldots,k\}$ and $a\in\N\setminus\{0\}$,
$$
P^{(n)}(a)\ =\ \dfrac{M^{(n)}_a}{2^{s(n)} a^{n-1}}
+ \dfrac{\xi^{(n)}_a}{2^{t(n)} a^n}
\qquad\textit{with $M^{(n)}_a\in\Z$ 
and $|\xi^{(n)}_a|\leq K_n$\ .}
$$
where, for $n\geq1$,
$s(n)=\dfrac{(n-1)(n-2)}{2}$, $t(n)=\dfrac{(n-1)n}{2}=s(n)+n-1$
and $K_n=K\prod_{i=1}^{i=n}(2^{i-1}+1)$.}
\medskip\\
{\it Proof.}
We argue by induction on $n$.
Case $n=1$ is as follows:
$$
\begin{array}{rclr}
P^{(1)}(a) &=&
\dfrac{\varphi(2a)-\varphi(a)+\left(\theta_{2a}-\theta_a\right)K}{a}
&\text{(cf. point (b) above)}
\medskip\\
&=& N_{2a,a} + \dfrac{\left(\theta_{2a}-\theta_a\right)K}{a}
&\text{(cf. point (a) above)}
\\
&=&M^{(1)}_a +\dfrac{\xi^{(1)}_a}{a}
&\text{where $M^{(1)}_a\in\Z$ and $|\xi^{(1)}_a|\leq K_1=2K$}
\end{array}
$$
Induction step $2\leq n\leq k$.
$$
\begin{array}{rcl}
P^{(n)}(a)&=&\dfrac{P^{(n-1)}(2a)-P^{(n-1)}(a)}{a}
\medskip\\
&=&
\dfrac{1}{a}\left(\dfrac{M^{(n-1)}_{2a}}{2^{s(n-1)} (2a)^{n-2}}
- \dfrac{M^{(n-1)}_a} {2^{s(n-1)} a^{n-2}}
+\dfrac{\xi^{(n-1)}_{2a}} {2^{t(n-1)} (2a)^{n-1}}
-\dfrac{\xi^{(n-1)}_a} {2^{t(n-1)} a^{n-1}}\right)
\medskip\\
&=&
\dfrac{M^{(n-1)}_{2a} - 2^{n-2} M^{(n-1)}_a}
         {2^{s(n-1)+n-2}\; a^{n-1}}
+\dfrac{\xi^{(n-1)}_{2a} - 2^{n-1} \xi^{(n-1)}_a}
           {2^{t(n-1)+n-1}\; a^n}
\medskip
\\
&=&\dfrac{M^{(n)}_a}{2^{s(n)}\; a^{n-1}}
+\dfrac{\xi^{(n)}_a}{2^{t(n)}\; a^n}
\hfill{\Box}
\end{array}
$$
with $M^{(n)}_a\in\Z$ 
and $|\xi^{(n)}_a|\leq K_n=(2^{n-1}+1)K_{n-1}$.
\smallskip

\noindent{\bf Claim 2.}{ \it
There exist integers $L,T\geq2$
such that, for every $a\in\N$, $a\geq1$,}
$$
\alpha_k \ =\
\lfloor\alpha_k\rfloor + \dfrac{N_a}{(La)^{k-1}} 
           + \dfrac{\eta_a}{(La)^k}
$$
for some $N_a\in\{0,1,\ldots,La-1\}$
and some real $\eta_a\in[-T, T]$.
\begin{proof}
Since $P^{(k)}(x)$ is the constant polynomial
$\alpha_k \ell$ where $\ell=\prod_{j=1}^{j=k}(2^j-1)$,
Claim 1 yields, writing the representation in base $La$ of $\alpha_k-\lfloor\alpha_k\rfloor$
for $L=\ell\; 2^{s(k)+k-1}$,
\begin{align*}
\alpha_k
\ =\
\dfrac{M^{(k)}_a}{\ell\; 2^{s(k)} a^{k-1}}
+ \dfrac{\xi^{(k)}_a}{\ell\; 2^{s(k)+k-1} a^k}
=\lfloor\alpha_k\rfloor + \dfrac{N_a}{(La)^{k-1}} + \dfrac{\eta_a}{(La)^k}
\end{align*}
where 
$T=L^{k-1} K_k$,
$\eta_a=L^{k-1} \xi^{(k)}_a$,
$N_a=L^{k-2}\; 2^{(k-1)^2} M^{(k)}_a$.
Clearly, $|\eta_a|\leq T$.
\end{proof}
%
%\noindent{\bf Claim 3.}{ \it
%If $k=1$ (i.e. the polynomial $P$ has degree $1$) 
%then $\alpha_1\in\Z$.}
%%
%\begin{proof}
%Claim 2 insures that
%$\alpha_1=\lfloor\alpha_1\rfloor+N_a+\dfrac{\eta_a}{La}$
%where $|\eta_a|\leq T$.
%For any $\varepsilon>0$, letting $a\geq\dfrac{T}{L\varepsilon}$,
%we see that the difference between $\alpha_1$ and the integer
%$\lfloor\alpha_1\rfloor+N_a$ is $\leq\varepsilon$.
%Since $\varepsilon$ is arbitrary small, this proves  $\alpha_1\in\Z$.
%\end{proof}
%
\noindent{\bf Claim 3.}{ \it
Suppose $k\geq2$.
For every $b\in\N$, $b>\dfrac{T}{L}$, the real $\alpha_k$ is in
\begin{equation*}
\dfrac{\Z}{(Lb)^\N}=\left\{\dfrac{p}{(Lb)^n} \mid p\in\Z,\; n\in\N\right\}\ .
\end{equation*}}
\begin{proof}
For any $s,i\in\N$ such that $s\geq2$ and $i\geq1$,
let $d(s,i)\in\{0,1,\ldots,s-1\}$ be the $i$-th digit
of the representation in base $s$
of the real $\alpha_k-\lfloor\alpha_k\rfloor$.
In case this real is in $\Z/s^\N$,
choose the representation which ends by an infinite tail of $0$'s.

Observe that the digit $d(s^p,k)$ (which is in $\{0,\ldots,s^p-1\}$)
is represented in base $s$ by the length $p$ sequence of digits
$d(s,i)$ (lying in $\{0,\ldots,s-1\}$) for $p(k-1)<i\leq pk$.

Claim 2 insures that if $a\in\N$ satisfies $La>T$ then, in base $La$,
the $k$-th digit $d(La,k)$ of $\alpha_k-\lfloor\alpha_k\rfloor$
is either in $\{0,\ldots,T\}$ (case $\eta_a\in[0,T]$)
or is in $\{La-T,\ldots,La-1\}$ (case $\eta_a\in[-T,0[$).

Let $a=L^{p-1}\;b^p$
where $p\geq2$ and $Lb>T$.
Then $La=(Lb)^p$ and the digit $d(La,k)$
which lies in $\{0,\ldots,T\}\cup\{La-T,\ldots,La-1\}$
is represented in base $Lb$ by one of the two length $p$ sequences
$$
00\ldots 0\lambda \quad\textit{(with $\lambda\in\{0,\ldots,T\}$)}
\quad,\quad
\delta\delta\ldots\delta\mu
\quad\textit{(with $\mu\in\{La-T,\ldots,La-1\}$)}
$$
In particular, going from base $La=(Lb)^p$ to base $Lb$,
the $p-1$ digits $d(Lb,i)$, for $p(k-1)<i<pk$,
are all $0$ or are all $\delta$.
Observe that, for $p>k+1$ we have $(p+1)(k-1)+1<pk-1$
hence the two sets $\{p(k-1)+1,\ldots,pk-1\}$
and $\{(p+1)(k-1)+1,\ldots,(p+1)k-1\}$ have non empty intersection.
As a consequence, the base $Lb$ digits 
$d(Lb,i)$, $i>(k+1)(k-1)+1=k^2$, are all $0$ or are all $\delta$.
Thus, the real $\alpha_k-\lfloor\alpha_k\rfloor$ is in
$\dfrac{\Z}{(Lb)^\N}$.
And so is the real $\alpha_k$.
\end{proof}
\noindent{\bf Claim 4.}{ \it
If $k\geq2$ then the real $\alpha_k$ is in $\Z$.}
\begin{proof}
Choose prime integers $b,c$ such that $b>c>\max(L,T/L)$.
Claim 3 insures that $\alpha_k$ is in both sets
$\dfrac{\Z}{(Lb)^\N}$ and $\dfrac{\Z}{(Lc)^\N}$.
Now, if $p,q\in\Z$ and $m,n\in\N$,
equality $\dfrac{p}{(Lb)^m}=\dfrac{q}{(Lc)^n}$
with $p,q\in\Z$ and $m,n\in\N$,
means $p L^n c^n = q L^m b^m$.
Since $b,c$ are distinct prime larger than $L$,
we see that $c^n$ divides $q$ and $b^m$ divides $p$.
In particular, $\dfrac{p}{(Lb)^m}\in\Z$.
Thus, the intersection of
$\dfrac{\Z}{(Lb)^\N}$ and $\dfrac{\Z}{(Lc)^\N}$
is $\Z$.
\end{proof}
This Claim finishes the proof of the theorem.
\end{proof}

%%%%%%%%%%%%%%%%%%%%%%%%%%%%
\subsection{Functions around the exponential functions}
\label{ss:exp not idr}
%%%%%%%%%%%%%%%%%%%%%%%%%%%%
%
%
We first apply the general negative result Theorem~\ref{thm:ae}.
\begin{theorem}\label{thm:uniformly close to alpha betaexp}
Let $\beta$ be a real number such that $\beta>1$. 
For almost all real numbers $\alpha$, the function $\N\to\RR$ such that
$n\mapsto\alpha\, \beta^n$ is not uniformly close to any function
$\N\to\Z$ having integral difference ratios.
\end{theorem}
\begin{proof}
Apply Theorem~\ref{thm:ae} with $\lambda_n=\beta^n$.
\end{proof}
For $\beta\in\N$, the above result holds for all $\alpha\neq0$
rather than almost all $\alpha$.
\begin{theorem}\label{thm:uniformly close to exp}
Let $\alpha$ be a non zero real number and $k\in\N\setminus\{0,1\}$. 
The function $\N\to\RR$ such that $n\mapsto\alpha\, k^n$ is not uniformly close to any function
$\N\to\Z$ having integral difference ratios.
\end{theorem}
Assuming  $f:\N\to\N$ is uniformly close to $\alpha k^x$,
to show that $f$ does not have integral difference ratios,
we apply the integral difference ratios assumption
to suitably chosen ordered pairs $\langle a,b\rangle$.

We first consider the case $\alpha\in\N\setminus\{0\}$
and show that, letting $\delta_x=f(x)-m k^x$,
 the sequence of deviations $(\delta_x)_{x\in\N}$
 is periodic (Lemma \ref{l:k-puissance-x}) and this yields a contradiction.
For $\alpha\in\RR\setminus\{0\}$,
using another family of ordered pairs $\langle a,b\rangle$,
we then prove that the base $k$ expansion of $\alpha$ is periodic
(Lemma \ref{l:alpha-rationnel}). 
The proof of Theorem \ref{thm:uniformly close to exp}
is then be easily concluded.

\begin{lemma}\label{l:k-puissance-x}
If $k\in\N\setminus\{0,1\}$ and
$f:\N\to\N$ is uniformly close to the function $x\mapsto m k^x$,
with $m\in\N$,
then $f$ does not have integral difference ratios.
\end{lemma}
\begin{proof}
Let us write $f(x)=mk^x+\delta_x$ with $\delta_x\in\N$ such that
$\delta_x<M$.
Let $\mu\in\N$ be such that $M<k^{\mu-1}-1$, hence 
\begin{equation}\label{deltanew}
\forall x\in \N\quad \delta_x=|f(x)-mk^x|<k^{\mu-1}-1
\end{equation}
Let us apply the integral difference ratios assumption with
$a-b = N(k^N-1)$ for some $N\geq\mu+2$.
Then $N(k^N-1)$ divides $f(a)-f(b)$, and
in particular,
\begin{eqnarray}\label{fa moins fb congru}
f(a)-f(b) &\equiv& 0  \pmod { (k^N-1)}\;.
\end{eqnarray}
Now, $c-1$ divides $c^d-1$ for all $c,d\geq1$.
Letting $c=k^N$ and $d=k^N-1$, we see that
\begin{eqnarray}\label{fk to the a-b congru}
k^{a-b}-1 \;=\;(k^N)^{k^N-1}-1 &\equiv& 0  \pmod { (k^N-1)}\;.
\end{eqnarray}
Since $f(a)-f(b)=mk^a-  mk^b   + (\delta_a-\delta_b) 
= mk^b(k^{a-b}-1) + (\delta_a-\delta_b) $,
equations \eqref{fa moins fb congru}  
and  \eqref{fk to the a-b congru}  yield   
\begin{eqnarray}\label{delta-congru}
      \delta_a-\delta_b&\equiv & 0  \pmod { (k^N-1)}
\\
\text{hence}\qquad\delta_a-\delta_b&=&0
\phantom{\pmod { (k^N-1)}}
\text{due to \eqref{deltanew} and $N\geq\mu+2$} .
\end{eqnarray}
The sequence $(\delta_x)_{x\in\N}$ is thus periodic, with period $ (k^N-1)$. 
Let $a$ be a multiple of $mk(k^N-1)$.
Since $k\geq2$, inequality \eqref{deltanew} yields
$|\delta_a-\delta_0| < k^\mu$.
Now, by \eqref{delta-congru},
$\delta_a-\delta_0$ is divisible by $k^N-1>k^\mu$.
Thus, $\delta_a-\delta_0 = 0$, so that $f(a)-f(0)= m(k^a-1)$.
This contradicts the integral difference ratios property because
$mk$ divides $a$ but does not divide $m(k^a-1)$.
\end{proof}

\begin{lemma}\label{l:alpha-rationnel}
If If $k\in\N\setminus\{0,1\}$ and
$f:\N\to\N$ has integral difference ratios and is uniformly close to the  function $x\mapsto \alpha k^x$, with $\alpha \in\RR$, then $\alpha $ is rational.
\end{lemma}
\begin{proof}
Let $M$ be such that $|f(x)-\alpha k^x|<M$ for all $x\in\N$.
For some $\mu\in\N$ we have $M< k^\mu$.
Let $\ell=k^{\mu+2}$ and $g(x)=f((\mu+2) x)$.
Then $g$ also has integral difference ratios and is uniformly close
to the  function $x\mapsto \alpha\, \ell^x$
and $|g(x)-\alpha\, \ell^x|<k^2=\ell/k^2\leq \ell/4$.

Thus, with no loss of generality,
we can reduce to the case $M=k/4$ with
\begin{equation}\label{delta}
f(x) = \lfloor\alpha k^x\rfloor + \delta_x
\text{\qquad where $\delta_x\in\N$
satisfies $|\delta_x|<k/2$.}
\end{equation}

We use the base $k$ expansion of integers and reals.
In case $\alpha$ is of the form $n/k^p$, with $p\in \N$, we systematically consider its infinite base $k$
expansion which ends with a tail of $0$'s and not a tail of $(k-1)$'s.
The (finite) base $k$ expansions of the integers
$\lfloor\alpha k^b\rfloor$ and $\lfloor\alpha k^a\rfloor$
are  related to the (infinite) base $k$ expansion
of the real $\alpha$.
If $w$ is a (finite or infinite) word on the alphabet
$\{0, 1,\cdots,k-1, \;.\;\}$,
we denote by $\overline{w}$ the integer or real having $w$ as
base $k$ expansion. Then
\begin{eqnarray}
\notag%\label{alpha in base k}
\alpha &=& \overline{t_0t_1\ldots t_p\;.\;t_{p+1}t_{p+2}\ldots}
\\\notag%\label{alpha k to c}
\lfloor\alpha k^c\rfloor
&=& \overline{t_0t_1\ldots t_{p+c}}\qquad\qquad\qquad
{\text {(with our convention on tails)}}
\\\notag
&\equiv& t_{p+c} \pmod  k
\\\label{eq:fa fb is}
\lfloor\alpha k^{a}\rfloor-\lfloor\alpha k^{b}\rfloor
+(\delta_{a}-\delta_{b})
&\equiv&t_{p+a} - t_{p+b} + (\delta_{a}-\delta_{b})\quad\pmod k
\end{eqnarray}
where the digits $t_i$'s, $i\in\N$, are in $\{0,1,\ldots,k-1\}$
and $|\delta_a-\delta_b|<k$.

Let $b\in\N$ and $a=b+k$.
The integral difference ratios assumption insures that
$$
\begin{array}{lrcl}
&f(a)-f(b)&\equiv&0\quad\pmod k
\\
\text{but}&
f(a)-f(b)&=&\lfloor\alpha k^{a}\rfloor-\lfloor\alpha k^{b}\rfloor
+(\delta_{a}-\delta_{b})
\\
\text{hence}&
t_{p+a} - t_{p+b} + (\delta_{a}-\delta_{b})
&\equiv&0\quad\pmod k
\end{array}
$$
Since 
$|t_{p+a} - t_{p+b} + (\delta_{a}-\delta_{b})|\leq2k-2$,
we see that
$t_{p+a} - t_{p+b} + (\delta_{a}-\delta_{b})\in\{-k,0,k\}$.
Recalling that $b$ is arbitrary and $a=b+k$,
this means that
\begin{equation}\label{eq:period}
\forall n\geq p\ \ \ 
\overline{t_{n}t_{n+1}} - \overline{t_{n+k}t_{n+1+k}}
\in\{-k,0,k\}
\end{equation}
To conclude, we argue by cases.
\\
{\it Case $t_n=t_{n+k}$ for some $n\geq p$.}
Then \eqref{eq:period} insures that
$\overline{t_{n}t_{n+1}} - \overline{t_{n+k}t_{n+1+k}}=0$
hence $t_{n+1}=t_{n+1+k}$
and, via an obvious induction, $t_m=t_{m+k}$ for all $m\geq n$.
In particular, $\alpha$ is eventually periodic with period $1$
hence $\alpha$ is rational.
\\
{\it Case $t_n\neq t_{n+k}$ for all $n\geq p$.}
First, \eqref{eq:period} insures that 
(either $t_n = t_{n+k}+1$ or $t_n +1 = t_{n+k}$)
and $t_{n+1}=t_{n+1+k}$, contradicting the assumption.
\end{proof}

{\it {Proof of Theorem~\ref{thm:uniformly close to exp}}}.
By Lemma \ref{l:alpha-rationnel},
if  $x\mapsto \alpha k^x$ is uniformly close to a function $f$
having integral difference ratios,
then $\alpha $ is rational. Let $\alpha =m/n$
with $m,n\in\N$. But $|f(x)-(m/n)  k^x|<M$  implies
$|nf(x)-m  k^x|<nM$, 
and $nf$ has integral difference ratios (by Proposition \ref{p:closure sum},
closure under sum) and $nM$ is a constant, 
hence $x\mapsto m  k^x$ is uniformly close to a function 
having integral difference ratios, contradicting Lemma \ref{l:k-puissance-x}.
\qed

%%%%%%%%%%%%%%%%%%%%%%%%%%%%
\subsection{Functions around the factorial function}
\label{ss:fact not idr}
%%%%%%%%%%%%%%%%%%%%%%%%%%%%
%
The following results make Theorem~\ref{thm:e factorial}
all the more unexpected.
\begin{proposition}\label{p:fact not idr}
If $a\in\Z\setminus\{0\}$ then the function $a^x\;x!$ does not have integral difference ratios.
\end{proposition}
\begin{proof}
Let $y+1> a$ be prime and $x=2y+1$
and observe that $x-y=y+1$ divides $x!$ hence also $a^x\,x!$
but does not divide $a^y\,y!$
hence does not divide $a^x\,x!-a^y\,y!$\,.
\end{proof}
Due to Theorem~\ref{thm:e factorial},
the following strengthening of Proposition~\ref{p:fact not idr}
for the case $a=1$ fails for $a\in\Z\setminus\{0,1\}$.
It also stresses that the one-point modification of
$\lfloor e \ x!\rfloor$ and $\lceil e \ x!\rceil$
in Theorem~\ref{thm:e factorial} is no accident.
\begin{proposition}\label{p:alpha fact not idr}
Let $\alpha$ be a non zero real number.
The functions
$x\mapsto \lfloor\alpha \ x!\rfloor$
and $x\mapsto \lceil\alpha \ x!\rceil$
do not have integral difference ratios.
\end{proposition}
\begin{proof}
We reduce to the case $\alpha>0$ since
$\lfloor-r\rfloor=-\lceil r\rceil$ for $r\in\RR$.
We consider the $\lfloor\ldots\rfloor$ case,
the $\lceil\ldots\rceil$ case being similar.
Arguing by contradiction, assume $\lfloor\alpha\times x!\rfloor$
has integral difference ratios.
Let $\theta_a\in[0,1[$ be such that
$\alpha\;a!=\lfloor \alpha\;a! \rfloor +\theta_a$.

First, we prove that if $\colon x\mapsto \lfloor\alpha \ x!\rfloor$ has integral difference ratios, then
$\alpha$ is a rational number.
Since $0!=1!=1$, applying the integral difference ratios property to
$a\in\N\setminus\{0\}$ and $b=0$ and $b=1$,
we see that both $a$ and $a-1$ divide
$\lfloor \alpha\;a!\rfloor - \lfloor \alpha\rfloor$
hence (since $a$ and $a-1$ are relatively prime) $a(a-1)$
divides $\lfloor \alpha\;a!\rfloor - \lfloor \alpha\rfloor$.
Thus, there exists $K_a \in\N$ such that
\begin{eqnarray}\notag
a(a-1)K_a &=& \lfloor \alpha\;a!\rfloor - \lfloor \alpha\;0!\rfloor
 \;=\;(\alpha\;a! -\theta_a)-(\alpha\;0! -\theta_0)
 \\\label{eq:alpha fact a plus delta}
&=&\alpha\;a! +\delta_a
\qquad\text{where $|\delta_a|\leq 2+\alpha$}
\\\label{eq:alpha Ka a-2}
\alpha&=&\dfrac{K_a}{(a-2)!}-\dfrac{\delta_a}{a!}
\\\label{eq:alpha Ka a-1}
\text{hence}\quad\alpha&=&\dfrac{K_a(a-1)}{(a-1)!}-\dfrac{\delta_a}{a!}
\\\label{eq:alpha Kaplus1}
\alpha&=&\dfrac{K_{a+1}}{(a-1)!}-\dfrac{\delta_{a+1}}{(a+1)!}
\quad\text{(replace $a$ by $a+1$ in \eqref{eq:alpha Ka a-2})}
\end{eqnarray}

Let $\round(\alpha,N)$ be the unique integer $x$ such that
$\alpha\in[\frac{x}{N}-\frac{1}{2N},\frac{x}{N}+\frac{1}{2N}[$.

Since $|\delta_a|,|\delta_{a+1}|$ are bounded by $2+\alpha$
then, for $a$ large enough,
\eqref{eq:alpha Ka a-2} insures that $\round(\alpha,(a-2)!)=K_a$
whereas \eqref{eq:alpha Ka a-1} and \eqref{eq:alpha Kaplus1}
insure that $\round(\alpha,(a-1)!)=K_a(a-1)=K_{a+1}$.
Thus, $K_a(a-1)=K_{a+1}$ hence
$\dfrac{K_a}{(a-2)!}=\dfrac{K_{a+1}}{(a-1)!}$
is a rational constant $r$ independent of $a$ for $a$ large enough.
%Let $r$ be the common rational value of the $\dfrac{K_a}{(a-2)!}$'s,
%for $a$ big enough.
Equations \eqref{eq:alpha fact a plus delta} and \eqref{eq:alpha Kaplus1}
insure that
$|\alpha-r| < (2+\alpha)/(a+1)!$ for all $a$ big enough; 
hence $|\alpha-r|$ is arbitrarily small and thus
 $\alpha=r$ is a rational number.

We can now get the wanted contradiction.
Let $\alpha=p/q$ where $p,q$ are relatively prime.
Let $a$ be such that $a-q$ is prime and $a-q>p\;q!$.
Since $a>q$, we have $\alpha\;a!=p(a!)/q\in\N$ 
hence $\lfloor\alpha\;a!\rfloor=\alpha\;a!$
and $a-q$ divides $\lfloor\alpha\;a!\rfloor$.

Also, $\alpha\;q!=p(q-1)!\in\N$ 
hence $\lfloor\alpha\;q!\rfloor=\alpha\;q!=p(q-1)!$. 
Since $a-q$ is prime and $a-q >p\;q!$, it cannot divide
$\lfloor\alpha\;q!\rfloor$.
Thus, $a-q$ cannot divide
$\lfloor\alpha\;a!\rfloor - \lfloor\alpha\;q!\rfloor$,
contradicting the integral difference ratios assumption.
\end{proof}
Theorem \ref{thm:factorial uniformly close} shows that the uniform closeness analog of
Corollary~\ref{cor:alpha factorial equivalent}  fails.
\begin{theorem}\label{thm:factorial uniformly close}
For $a\in\RR\setminus\{0,1\}$,
let $X_a$ be the set of real numbers $\alpha$ such that
the  map $x\mapsto\alpha\,a^x\;x!$ from $\N$  to $\RR$
is uniformly close to some function having integral difference ratios.
\\
1. For every $a\in\RR\setminus\{0\}$,
the set $X_a$ has Lebesgue measure zero.
\\
2. If $a\in\Z\setminus\{0\}$,
the set $X_a$ contains $e^{1/a}\,\Z$
but misses every non null rational number.
\end{theorem}
\begin{proof}
1. Applied with $\lambda_n=a^n\;n!$, Theorem~\ref{thm:ae} 
insures that $X_a$ has measure zero.
%%%%%%%%%%%%%
\smallskip\\
2. Inclusion $X_a\supseteq e^{1/a}\N$, for $a\in\Z\setminus\{0\}$,
is a consequence of Corollary~\ref{thm:e factorialbis}.

By way of contradiction, assume some non null rational number
$\alpha$ is in $X_a$.
Let $\alpha=p/q$ where $p\in\Z$ and $q\in\N\setminus\{0,1\}$.
Let $f:\N\to\Z$ and $M\in\N$ be such that
$f$ has integral difference ratios and
$\left|f(x)-(p/q)\,a^x\;x!\right|< M$ for all $x\in\N$.
Thus,
$$
q\,f(x)=p\,a^x\;x! + \varepsilon_x\,q\,M
\quad\text{with $|\varepsilon_x|< 1$.}
$$
Choose $y=2qM$ and $x=p\,a^y\;y!+y+1$.
Every divisor of $q$ divides $y$ and $x-y-1$
hence does not divide $x-y$. Thus, $q$ and $x-y$ are coprime.
Since $q,x-y<x$, we see that $q (x-y)$ divides $x!$ hence
also $p\,a^x\;x!$.
Since $q (x-y)$ divides $q\,(f(x)-f(y))$,
we see that $q\,(x-y)$ divides
\begin{eqnarray}\notag
p\,a^x\;x! - q\, (f(x)-f(y)) 
&=&p\,a^x\;x! - \left(p a^x\;x! +\varepsilon_x\,q\, M\right)
+ \left(p\,a^y\;y! +\varepsilon_y\,q\, M\right)
\\\label{eq:pax - fx fy}
&=&p\,a^y\;y! + \ell_{x,y}
\end{eqnarray}
where $|\varepsilon_x|, |\varepsilon_1| < 1$
and $\left|\ell_{x,y}\right| <2qM=y$, 
whence $p\,a^y\;y!+\ell_{x,y}  <p\,a^y\;y!+y$. 
%Finally,\change{unless $\ell_{x,y}=1$ car $x-y=p\,a^y\;y!+1$}
%$$
%%\ <\ p\,a^y\;y!+y\ =\ x-1
%$$
Then $q(x-y)=q(pa^yy! +1)>2pa^yy!$  cannot divide
$p\,a^y\;y! + \ell_{x,y}<p\,a^y\;y! + y<2pa^yy!$. Contradiction.
\end{proof}
How complex are the real numbers in the set $X_a$
of Theorem~\ref{thm:factorial uniformly close}?
First, we recall the notions of irrationality measure and
Liouville numbers.
\begin{definition}
1. The {\rm irrationality exponent} of a real number $\alpha$ is the
supremum of all $\mu\in\RR^+$ such that
the approximation\ $|\alpha - (p/q)| < 1/q^\mu$\
holds for infinitely many rational numbers $p/q$.
\\
2. A real number is {\rm Liouville} if its irrationality exponent is infinite.
\end{definition}
\begin{proposition}
1. Rational numbers have irrationality exponent $1$.
\\
2. Irrational numbers have irrationality exponent at least $2$.
\\
3. (Roth, 1955) All irrational algebraic numbers have
irrationality exponent $2$.
\\
4. (Khinchin, 1924, cf. \cite{Shidlovskii1989} p.17) Almost all real numbers have irrationality exponent $2$.
\end{proposition}
For $a\in\Z\setminus\{0\}$,
the sole numbers known to be in the set $X_a$
of Theorem~\ref{thm:factorial uniformly close}
are those in the set $e^{1/a}\,\Z$.
It turns out that all have irrationality exponent equal to $2$.
\begin{proposition}
All numbers $s\,e^{1/a}$, $a,s\in\Z\setminus\{0\}$,
have irrationality exponent $2$.
In particular, though they are transcendental,
they are not Liouville.
\end{proposition}
\begin{proof}
It obviously suffices to consider $s=1$.
First, we consider the case $a\in\N\setminus\{0\}$.

The continued fraction expansion $[a_0;a_1,a_2,\ldots]$
of $e^{1/a}$ was computed by Euler
(cf. \cite{Waldschmidt2008} page 10):
$$
\begin{array}{rclcl}
e&=&[2,1,2,1,1,4,1,1,6,1,1,\ldots]
&=&[2;\overline{1,2m,1}]_{m\in\N}
\\
\text{for $a\geq2,$}\ \ \
e^{1/a}&=&[1;a,1,1,3a,1,1,5a,\ldots]
&=&[\overline{1,(2m+1)a,1}]_{m\in\N}
\end{array}
$$
Let $p_n/q_n$ be its $n$-th convergent.
As a general result for all irrational numbers
(cf. \cite{khinchin1964} Theorems 6, 12 and 9 \& 13),
we have
\begin{eqnarray}\label{eq:qn1 over qn}
\dfrac{q_{n+1}}{q_n}&=&[a_{n+1};a_n,\ldots,a_1]
\ <\ a_{n+1}+1
\\\label{eq:qn 2 to the n}
q_n&\leq&2^{(n-1/2)}
\\\label{eq:rat approx}
\dfrac{1}{q_n\,(q_n+q_{n+1})} 
&<&\left| e^{1/a} - \dfrac{p_n}{q_n}\right|
\ <\ \dfrac{1}{q_n\,q_{n+1}}
\end{eqnarray}
The above Euler formulas show that,
the continued fraction expansion $e^{1/a}$ is such that,
$a_{n+1}\leq n-1$ for $n\geq3$.
Thanks to \eqref{eq:qn1 over qn}, we get $q_{n+1}\leq n\,q_n$.
Using \eqref{eq:qn 2 to the n}, for any $\varepsilon>0$, 
this yields
$q_n\,(q_n+q_{n+1})\leq (n+1)\,q_n^2Ê\leq q_n^{2+\varepsilon}$
for $n$ large enough.
Reporting in \eqref{eq:rat approx}, we finally have 
$$
\dfrac{1}{q_n^{2+\varepsilon}}
<\left| e^{1/a} - \dfrac{p_n}{q_n}\right|
<\dfrac{1}{q_n^2}
$$
which proves that the irrationality exponent of $e^{1/a}$
is $2$.
The argument also applies to $e^{-1/a}=1/e^{1/a}$
since its continued fraction expansion 
is $[0;a_0,a_1,a_2,\ldots]$.
\end{proof}
\noindent
{\bf Open problem.}
Does the set $X_a$ of Theorem~\ref{thm:factorial uniformly close} contain numbers with irrationality measure other than $2$?
\smallskip

We can only prove that it misses a subfamily of Liouville numbers.
\begin{definition}
Let $\theta:\N\to\N$ be such that,
for all $q\in\N$, the map
$n\mapsto\theta(q,n)$ goes to infinity.
A real number $\alpha$ is $\theta$-Liouville if, for all $n\in\N$,
there exists a rational number $p/q$ such that
\begin{equation}\label{eq:theta liouville}
\left|\alpha - \dfrac{p}{q}\right| < \dfrac{1}{\theta(q,n)}\ .
\end{equation}
\end{definition}
\begin{theorem}\label{thm:factorial uniformly close liouville}
Let $\expp:\N\to\N$ and $\theta:\N\to\N$
be such that $\expp(x)=x^{x^x}$ and $\theta(q,n)=\expp(q\,n)$.
\\
1. All $\theta$-Liouville numbers are Liouville numbers.
\\
2. For $a\in\Z\setminus\{0,1\}$,
the set $X_a$ misses all $\theta$-Liouville numbers.
\end{theorem}
\begin{remark}
Recall (cf. \cite{khinchin1964} Theorem 22,
or \cite{Shidlovskii1989} Theorem 1 p.12)
that, for any $\theta$, there exist $\theta$-Liouville real numbers.
For the above $\theta$, an example is
$\sum_{n\in\N}\dfrac{1}{a_n}$
where $a_0=2$ and $a_{n+1}=2^{a_n+\exp3(n\,a_n)}$.
\end{remark}
\begin{proof}[Proof of Theorem~\ref{thm:factorial uniformly close}.]
1. Obvious since $\theta(q,n)\geq q^n$.
\\
2. Let $\alpha$ be $\theta$-Liouville.
With no loss of generality, we can assume $\alpha>0$.
Let $C\in\N$ be such that $\alpha< C$. Then
\begin{equation}\label{eq:liouville}
\forall n\in\N\ \ \exists p\in\Z \ \ \exists q\in\N\setminus\{0\}
\ \ \left(\left|\alpha - \dfrac{p}{q}\right| 
< \dfrac{1}{q\,\expp(nq)} \text{ and }p<Aq\right)\ .
\end{equation}
By way of contradiction, suppose $\alpha$ is in $X_a$
and let $f:\N\to\Z$ and $M\in\N$, $M\geq1$, be such that
$f$ has integral difference ratios and
$\left|f(x)-\alpha\,a^x\;x!\right|< M$ for all $x\in\N$.
We argue in a way similar to that in the proof of
Theorem~\ref{thm:factorial uniformly close}.

Choose $y=5qM+C$ and $x=p\,a^y\;y!+y+1$.
Observe that, for $n$ large enough, we have
\begin{equation}\label{eq:for n large}
a^x\;x!+a^y\;y! \, \leq qM\,\expp(nq)
\end{equation}
In fact, by Stirling formula, there exist $B,C$
(depending only on $C$ and $M$) such that, for $n>B$,
$$
x\leq (Aq)^{Aq}
\quad,\quad
a^x\;x!+a^y\;y!\leq (Bq)^{(Bq)^{Bq}}\ <\ qM\,\expp(nq)
$$
Choose such an $n\in\N$ and let $p,q$ be as in \eqref{eq:liouville}.
Then, for all $z\in\N$,
\begin{equation}\label{eq:q fx}
q\;f(z)\ =\ q\,\alpha\,a^z\;z! +\varepsilon\,q\, M 
\ =\ p\,a^z\;z! +\delta\,\dfrac{a^z\;z!}{\expp(nq)}
+\varepsilon\,q\, M
\end{equation}
where $|\varepsilon|,|\delta| <1$.
Arguing as for the case $\alpha$ rational,
we see that $q\,(x-y)$ divides $p\,a^x\;x! - q\, (f(x)-f(y))$.
Thus, to get a contradiction, it suffices to show 
\begin{equation}\label{eq:to be proved}
0\ <\ p\,a^x\;x! - q\, (f(x)-f(y)) \ <\ x-1\ .
\end{equation}
Now, \eqref{eq:q fx} insures that
\begin{equation}
p\,a^x\;x! - q\, (f(x)-f(y)) 
\;=\; p\,a^y\;y! + \ell_{x,y} +
\Delta\,\dfrac{a^x\;x! + a^y\;y!}{\expp(nq)}
\end{equation}
Using \eqref{eq:for n large}
and inequalities $|\ell_{x,y}| <2qM$, and $|\Delta|<1$,
we see that 
\begin{equation*}
\left|\ell_{x,y} 
+ \Delta\,\dfrac{a^x\;x! + a^y\;y!}{\expp(nq)}\right| \leq y
\end{equation*}
hence
$$
0\ <\ p\,a^y\;y!- y
\ <\ p\,a^y\;y!+\ell_{x,y} 
+ \Delta\,\dfrac{a^x\;x! + a^y\;y!}{\expp(nq)}
\ <\ p\,a^y\;y!+y\ =\ x-1\ .
$$
This gives the wanted inequality and concludes the proof.
\end{proof}

%%%%%%%%%%%%%%%%%%%%%%%%%%%%%%%%%%
%%%%%%%%%%%%%%%%%%%%%%%%%%%%%%%%%%
%%%%%%%%%%%%%%%%%%%%%%%%%%%%%%%%%%
%%%%%%%%%%%%%%%%%%%%%%%%%%%%%%%%%%
%%%%%%%%%%%%%%%%%%%%%%%%%%%%%%%%%%
%%%%%%%%%%%%%%%%%%%%%%%%%%%%%%%%%%
%%%%%%%%%%%%%%%%%%%%%%%%%%%%%%%%%%
%%%%%%%%%%%%%%%%%%%%%%%%%%%%%%%%%%
%%%%%%%%%%%%%%%%%%%%%%%%%%%%%%%%%%
%%%%%%%%%%%%%%%%%%%%%%%%%%%%%%%%%%
%%%%%%%%%%%%%%%%%%%%%%%%%%%%%%%%%%
%%%%%%%%%%%%%%%%%%%%%%%%%%%%%%%%%%

\end{document}